\documentclass{llncs}
\usepackage{etex}

\newif\iffigures
\figurestrue % uncomment if you want figures
\usepackage{etex}
\usepackage{setspace}
\usepackage{amsmath} % assumes amsmath package installed
\usepackage{amssymb,nicefrac}
\usepackage[dvips]{epsfig}
\usepackage[dvips]{graphicx}
\usepackage{textcomp}
\usepackage{fancyhdr}
\usepackage{color}
\usepackage{pstricks}
\usepackage{pst-node}
\usepackage{dsfont}
\usepackage{mathrsfs}
\usepackage{sgame}
\usepackage{float}
\usepackage[T1]{fontenc}
\usepackage{aurical}
\usepackage{tikz,pgfplots,verbatim}
\usetikzlibrary{spy,arrows,backgrounds,plotmarks}
\iffigures
\usepgfplotslibrary{external} 
\newtheorem{assumption}{Assumption}[section]
\newcommand{\tr}{^{\mathrm T}}

\newcommand{\df}{\doteq}
\newcommand{\magn}[1]{\left\vert #1 \right\vert}

\newcommand{\vp}{\mathbf{p}}
\newcommand{\vpq}{\mathbf{p}\alpha}
\newcommand{\vv}{\mathbf{v}}
\newcommand{\vS}{\mathbf{S}}
\newcommand{\cM}{\mathcal{M}}
\newcommand{\cP}{\mathcal{P}}

\newcommand{\Spq}{S_{\vp{\alpha}}}
\newcommand{\vw}{\mathbf{w}}

\newcommand{\ve}{\mathbf{e}}
\newcommand{\Ap}{\mathcal{A}(\vp)}
\newcommand{\cA}{\mathcal{A}}

\begin{document}

\title{Supervisory Output Prediction for Bilinear Systems by Reinforcement Learning\thanks{This work is part of the mpcEnergy project which is supported within the program Regionale Wettbewerbsf\"{a}higkeit O\"{O} 2007-2013 by the European Fund for Regional Development as well as the State of Upper Austria. The research reported in this article has been (partly) supported by the Austrian Ministry for Transport, Innovation and Technology, the Federal Ministry of Science, Research and Economy, and the Province of Upper Austria in the frame of the COMET center SCCH.}\thanks{An earlier version of part of this paper appeared
in \cite{ChasparisNatschlaeger14_MSC}.}}

%\title{Efficient Dynamic Pinning of Parallelized Applications by Distributed Reinforcement Learning\thanks{This work has been partially supported by the European Union grant EU H2020-ICT-2014-1 project RePhrase (No. 644235).}%\thanks{Grants or other notes
%%about the article that should go on the front page should be
%%placed here. General acknowledgments should be placed at the end of the article.}
%}
%% \subtitle{Do you have a subtitle?\\ If so, write it here}

\titlerunning{Supervisory Output Prediction for Bilinear Systems by Reinforcement Learning}        % if too long for running head

\author{Georgios C. Chasparis \and Thomas Natschl\"{a}ger
}


\institute{Software Competence Center Hagenberg GmbH, Softwarepark 21, A-4232 Hagenberg, Austria \\
\email{\{georgios.chasparis,thomas.natschlaeger\}@scch.at} %, \\ WWW home page: \texttt{http://www.scch.at
}

%\author[1,*]{Georgios C. Chasparis}
%\author[1]{Thomas Natschl\"{a}ger}
%\affil{Software Competence Center Hagenberg GmbH, Softwarepark 21, A-4232 Hagenberg, Austria}
%\affil[*]{georgios.chasparis@scch.at}

\maketitle

\begin{abstract}
Online output prediction is an indispensable part of any model predictive control implementation, especially when simplifications of the underlying physical model have been considered and/or the operating conditions change quite often. Furthermore, the selection of an output prediction model is strongly related to the data available, while designing/altering the data collection process may not be an option. Thus, in several scenarios, selecting the most appropriate prediction model needs to be performed during runtime. To this end, this paper introduces a supervisory output prediction scheme, tailored specifically for input-output stable bilinear systems, that intends on automating the process of selecting the most appropriate prediction model during runtime. The selection process is based upon a reinforcement-learning scheme, where prediction models are selected according to their prior prediction performance. An additional selection process is concerned with appropriately partitioning the control-inputs' domain in order to also allow for switched-system approximations of the original bilinear dynamics. We show analytically that the proposed scheme converges (in probability) to the best model and partition. We finally demonstrate these properties through simulations of temperature prediction in residential buildings.
\end{abstract}

%%%%%%%%%%%%%%%%%%%%%%%%%%%%%%%%%%%%%%%%%%%%%%%%%%%%%%%%%%%%%%%%%%%%%%%%%%%%%%%%%%%%%%%%
% INTRODUCTION
%%%%%%%%%%%%%%%%%%%%%%%%%%%%%%%%%%%%%%%%%%%%%%%%%%%%%%%%%%%%%%%%%%%%%%%%%%%%%%%%%%%%%%%
\section{Introduction}	\label{sec:Introduction}

Bilinear systems play a significant role in modeling several physical and engineering processes, such as population dynamics of biological species \cite{Mohler70}, % modeling of nonlinear tracking and 
communication systems \cite{Pardalos08}, chemical processes \cite{Favoreel99}, and thermal dynamics in buildings \cite{ChasparisNatschlaeger16}. Due to their numerous applications, system identification for bilinear systems has attracted considerable attention, including: a) \emph{prediction-error methods} % through the formulation of a nonlinear basis for regression or Volterra series expansion 
\cite{Inagaki84,Schrempf01}, b) \emph{maximum likelihood methods} % employed on state-space formulations 
\cite{Fnaiech87,Verdult02,Gibson05}, and c) \emph{subspace identification methods} \cite{Favoreel99,Verdult02b}, % or \emph{iterative subspace identification methods} 
\cite{LoperDosSantos09}. 

When predictions of a system's response needs to be provided \emph{online} (as in a model predictive control implementations), the identification process selected, whether it is (a), (b) or (c), will operate frequently during runtime. This choice is supported by the fact that i) \emph{changes in the operating conditions of the system may degrade the performance of the identification process}, and ii) \emph{designing the data collection process may not be an option}. 
% The need for online prediction is further supported by the fact that designing the data collection process, as studied in \cite{Sontag09}, may not be an option. Instead, data may only be collected during normal operation of the system. Such restriction may have an impact on the performance of any prediction model, given that different models exhibit different training times. 
Providing that predictions are requested often, the following questions naturally arise: a) \emph{Is there a prediction model that better approximates the system's output during runtime?} and b) \emph{Will a switching strategy between alternative prediction models provide a better performance in the long-run?}

The goal of this paper is to automate the process for providing answers to the above questions through a \emph{supervisory approach for online prediction}. Such supervisory scheme does not depend on the specific class of the prediction models used (e.g., Volterra series expansions or subspace methods), instead, switches between the available ones and selects the most appropriate during runtime. The objective is to \emph{always balance between short training times and good predictions under any possible operating conditions.}

The supervisory process will be specifically tailored for input-output stable bilinear systems. In particular, it comprises two parallel decision processes that run periodically on a fixed period. The \emph{first} one is concerned with the selection of an appropriate partition of the input(s) domain, %\emph{exploiting the fact that bilinear models can be approximated well by linear models when the control inputs remain almost constant}. 
while the second one is concerned with the selection of the most appropriate prediction model for each one of the resulting partition sets. The resulting prediction model constitutes a \emph{switched} system, where a different prediction model applies to each partition set of the input(s) domain. \emph{We show analytically that the proposed scheme is adaptive and robust to changes in the performance of the prediction models, while convergence is attained (in probability) to the best partition and model.}

The idea of representing a nonlinear system by a switched system fits to the linear parameter-varying (LPV) systems framework \cite{ShammaAthans91}, and the piecewise-linear approach of \cite{Sontag81}. It has also been considered for identification of nonlinear systems as in the off-line multiple-model approaches of \cite{Venkat03,Thiaw07,Orjuela13}, where a family of models is considered, each of which applies with different weight to each cluster of the state space. Contrary to this line of research, in this paper a) the weights with which each model participates in a subregion of the inputs domain will emerge \emph{dynamically}, depending on its performance during runtime, % rather than off-line based on a predefined clustering division process, 
and b) the proposed approach is not necessarily restricted to linear models as in \cite{Orjuela13}. Lastly, the idea of supervisory processes in controls is also not new, including, for example, the switching supervisory controller \cite{Alshyoukh09} for stabilizing unstable linear systems. The intention though here is different since we are concerned with online prediction for bilinear systems. The paper extends prior work of the authors \cite{ChasparisNatschlaeger14_MSC} by admitting weaker assumptions in the derivation of the convergence properties.

% Summarizing, the contributions of the paper are as follows: (a) it introduces a supervisory (switched-system) online prediction scheme specifically tailored for bilinear systems, that balances between training times and performances of alternative prediction models, (b) the weights assigned to each model is decided during runtime based on each model's performance, rather than offline, and (c) the supervisory process exhibits convergence (in probability) to the ``currently'' best partition scheme and models, thus exhibits adaptivity and robustness to possible variations in the performance of the models. 

In the remainder of the paper, Section~\ref{sec:FrameworkObjective} describes the framework and objective. Section~\ref{sec:SwitchingSupervisoryID} introduces the proposed approach of the switching supervisory prediction scheme for bilinear systems. Section~\ref{sec:Convergence} provides the convergence properties of the supervisory scheme, including an evaluation of the supervisory approach through simulations of the indoor temperature in residential buildings. Section~\ref{sec:TechnicalDerivation} provides the technical details of the convergence analysis. Finally, Section~\ref{sec:Conclusions} presents concluding remarks.

\emph{Notation:} Throughout the paper, we use the notation:
\begin{itemize}
% \item for $\mathbf{x}\in\mathbb{R}^{n}$ and $\delta>0$, define the $\delta$-neighborhood of $\mathbf{x}$ as: $\mathcal{N}_{\delta}(\mathbf{x}) \df \{\mathbf{y}\in\mathbb{R}^{n}: \|\mathbf{x}-\mathbf{y}\|_2 \leq \delta\},$ where $\|\cdot\|_2$ denotes the Euclidean distance;
 \item for $\mathbf{x}\in\mathbb{R}^{n}$, $\mathbf{x}[j]$ denotes its $j$th entry, $j=1,...,n$;
 \item for $A\subset\mathbb{R}^{n}$ and $\delta>0$, define the $\delta$-neighborhood of $A$ as:
 $$\mathcal{N}_{\delta}(A) \df \{\mathbf{y}\in\mathbb{R}^{n}:\inf_{\mathbf{x}\in{A}}\|\mathbf{x}-\mathbf{y}\|_2\leq\delta\},$$ where $\|\cdot\|_2$ denotes the Euclidean distance;
% \item $\ve_{j}$ is the \emph{unit vector} of appropriate dimension, with 
% $$\ve_{j}[i] \df \begin{cases} 1, & \mbox{if } i=j \\
% 0, & \mbox{else}; 
% \end{cases}$$
% \item $\mathbb{I}_{\mathbf{x}\in{A}}$ denotes the \emph{index function}, for some set $A$, i.e.,
% \begin{equation*}
%   \mathbb{I}_{x\in{A}} \df 
%   \begin{cases}
%     1, & \mbox{ if } x\in{A}\\
%     0, & \mbox{ else }
%   \end{cases}
% \end{equation*}
 \item $\mathbf{\Delta}(n)\subset\mathbb{R}^{n}$ denotes the \textit{probability simplex} in $\mathbb{R}^{n}$ defined as follows: $$\mathbf{\Delta}(n)\df\Big\{\mathbf{x}=(x_1,...,x_n)\in\mathbb{R}^{n}_+:\sum_{i=1}^{n}x_i = 1\Big\};$$
% \item for a vector $\mathbf{x}\in\mathbf{\Delta}(n)$, let ${\rm rand}_{\mathbf{x}}[a_1,...,a_n]$ denote the random selection of an element of the set $\{a_1,...,a_n\}$ according to the distribution $\mathbf{x}$;
 \item ${\rm rand}_{\rm unif}[a_1,...,a_n]$ denotes the random selection of an element of the set $\{a_1,...,a_n\}$ according to the finite uniform distribution $\mathbf{x}=(\nicefrac{1}{n},...,\nicefrac{1}{n})$;
% % \item for a finite set $\{x_1,...,x_n\}$ in $\mathbb{R}$, ${\rm col}\{x_1,...,x_n\}$ denotes the column vector $(x_1,...,x_n)\in\mathbb{R}^{n}$, ${\rm row}\{x_1,...,x_n\}$ denotes the corresponding row vector and ${\rm diag}\{x_1,...,x_n\}$ denotes the diagonal matrix with diagonal entries $\{x_1,...,x_n\}$;
 \item for a finite set $A$, $\magn{A}$ denotes its cardinality.
% \item $\{a_i\}_{i=0}^{n}$ denotes the finite sequence $\{a_0,a_1,...,a_{n}\}$.
\end{itemize}

%%%%%%%%%%%%%%%%%%%%%%%%%%%%%%%%%%%%%%%%%%%%%%%%%%%%%%%%%%%%%%%%%%%%%%%%%%%%%%%%%%%%%%%%%%%%%%%%%%%%%%%%%%%%%%%%%%%%%%%%%
% SYSTEM MODELING
%%%%%%%%%%%%%%%%%%%%%%%%%%%%%%%%%%%%%%%%%%%%%%%%%%%%%%%%%%%%%%%%%%%%%%%%%%%%%%%%%%%%%%%%%%%%%%%%%%%%%%%%%%%%%%%%%%%%%%%%%
\section{Framework \& Objective}		\label{sec:FrameworkObjective}

\subsection{Framework}

We consider the problem of \emph{online output prediction for bilinear systems} that admit the generic state-space form:
\begin{eqnarray}	\label{eq:BilinearDynamics}
\dot{\mathbf{x}} 	& = & \mathbf{A} \mathbf{x} + \sum_{i=1}^{m}u_i\left(\mathbf{B}_i\mathbf{x} + \mathbf{D}_i\mathbf{d}\right) + \mathbf{D} \mathbf{d}, \cr
\mathbf{y} 		& = & \mathbf{C} \mathbf{x}
\end{eqnarray}
where $\mathbf{x}\in\mathcal{X}$ is the state vector, $u_i\in \mathcal{U}_i$, $i=1,...,m$, are the control inputs,  $\mathbf{d}\in\mathcal{D}$ is the disturbance vector, and $\mathbf{y}\in\mathcal{Y}$ is the output vector. Furthermore, the matrices $\mathbf{A},\mathbf{B}_i$, $i=1,...,m$, $\mathbf{C}$, $\mathbf{D}$ and $\mathbf{D}_i$, $i=1,...,m$, are constant real matrices of appropriate dimension. The domain of a control input $\mathcal{U}_i$, $i=1,...,m$, will be assumed to be a \emph{bounded closed subset} of $\mathbb{R}$. Lastly, the Cartesian product of $\{\mathcal{U}_i\}$ will be denoted by $\mathcal{U}\df \mathcal{U}_1\times\ldots\times\mathcal{U}_m$. 

We will be concerned with systems of the form (\ref{eq:BilinearDynamics}) that are \textit{bounded-input bounded-output stable}, and thus stabilization will not be a concern in this paper. Dynamics~(\ref{eq:BilinearDynamics}) provide an assumed approximation of the physical system without taking into account possible (noise) perturbations in the dynamics or in the measurements. Such perturbations can be considered when designing models for output prediction.

\subsection{Questions}		\label{sec:Questions}

Input-output pairs are recorded periodically at time instances $0,T_s,2T_s,...$, for some \emph{sampling period} $T_s>0$, briefly denoted by $\tau_0,\tau_1,\tau_2,...$, respectively. At any $\tau_j$, $j=0,1,...$, the updated history of input-output pairs is: $$H_{j}\df\{(\mathbf{u}(\tau_0),\mathbf{y}(\tau_0)),(\mathbf{u}(\tau_1),\mathbf{y}(\tau_1)),...,(\mathbf{u}(\tau_j),\mathbf{y}(\tau_j))\}.$$ Let also $\mathcal{H}_{j}$ be the family of such histories up to the sampling instance $\tau_j$. 

The problem of output prediction at time $\tau_j$ translates into the formulation of an estimate $\hat{\mathbf{y}}(\tau_{j+1})$ of the output for the next sampling instance, $\mathbf{y}(\tau_{j+1})$. Such estimates are formulated using a mapping of the form $S:\mathcal{H}_{j}\to\mathcal{Y}$, such that $S(H_{j})=\hat{\mathbf{y}}(\tau_{j+1})$.

\emph{\textbf{For the remainder of the paper}}, we assume that a set $\cM$ of user-defined prediction models is available to formulate predictions. Ideally, \emph{we would like to find which prediction model $S_j^*$ out of the available set $\cM$, would be the one that minimizes the prediction error over the next $N$ sampling instances}. In other words, we would like to solve the following optimization problem: 
%\begin{question}	\label{Q:Question1}
%For any sampling instance $\tau_j=jT_s$, what should be the prediction model $S_j^*\in\cM$ so that the prediction accuracy of the output (\ref{eq:BilinearDynamics}), over the next $N$ sampling instances, is minimized, i.e., what is
\begin{equation}	\label{eq:Question1}
S_j^*=\arg\min_{S\in\cM}\sum_{q=0}^{N-1}\left\|\mathbf{y}(\tau_{j+q+1})-S(H_{j+q})\right\|_2^{2}.
\end{equation}
%\end{question}

Of course, the above optimization is not well-posed, since the actual measurements $\mathbf{y}(\tau_{j+1})$,..., $\mathbf{y}(\tau_{j+N})$, over the next $N$ time instances, are not known at time $\tau_j$. \emph{In practice, we may only assess the performance of a model over the next $N$ instances, by evaluating its performance over the prior history, i.e., using only the realized $H_j$.} We should expect that a model that performed well over $H_j$, will continue to perform well for the upcoming $N$ time instances $\tau_{j+1}$,...,$\tau_{j+N}$. However, this is simply a hypothesis that needs investigation.

Instead, let us consider the following optimization, 
\begin{equation}	\label{eq:InitialOptimizationProblem}
\hat{S}_j = \arg\min_{S\in\cM} \sum_{q\in Q_j\subseteq\{-j,...,-1\}} \left\| \mathbf{y}(\tau_{j+q+1}) - S(H_{j+q})\right\|^2_2,
\end{equation}
where we evaluate the performance of the models over a sample $Q_j$ of only \emph{prior} observations in $H_j$. Ideally, we would like $\hat{S}_j\equiv S_{j}^*$ at all times, so that by evaluating the prediction performance over prior data, we get the best model for the next $N$ time instances. Thus, naturally, the following question emerges: 
\begin{question}	\label{Q:Question2}
Under which conditions (i.e., prior history $H_j$ and sample set $Q_j$) shall we expect that $\hat{S}_j \equiv S_{j}^*$ for all instances $j=0,1,...$?
\end{question}
Essentially, Question~2.2 asks whether it is reasonable to assume that a model that provided the best predictions so far (model $\hat{S}_j$) will still be the best model for the next evaluation period (which corresponds to model $S_{j}^*$). This is definitely not a trivial question to answer mainly due to the \emph{varying operating conditions}. When the underlying physical system is nonlinear, variations in the operating conditions may result in poor performance of several prediction models in $\cM$ (e.g., due to the lack of prior data from such operating conditions).

Assuming that Question~\ref{Q:Question2} can be answered (i.e., we may indeed use prior performances as a valid assessment criterion over next performances), the optimization (\ref{eq:InitialOptimizationProblem}) may not necessarily correspond to a computationally efficient optimization problem. That is primarily due to a) all models in $\cM$ need to be evaluated over the prior history $H_j$, and b) all models in $\cM$ need to be retrained as soon as new input-output pairs become available. These remarks give rise to the following question:
\begin{question}	\label{Q:Question3}
Alternatively to performing an exhaustive search for computing the best model in $\cM$, as the optimization problem~(\ref{eq:InitialOptimizationProblem}) dictates, are there computationally efficient methods for computing/updating $\hat{S}_{j}$ for each $j=0,1,...$?
\end{question}

This is a reasonable question, since, in several applications, the amount of data available continuously grow with time. Thus, it may be prohibited to evaluate/retrain all available models in $\cM$ at every evaluation instance $\tau_j$. In this case, it would be desirable that we provide a pattern over which models are evaluated/trained that does not require an exhaustive search approach as (\ref{eq:InitialOptimizationProblem}) dictates.

%\begin{enumerate}
%\item (\textit{Model variations.}) The prediction models available for prediction $\cM$ are continuously trained as new input-output measurements are gathered. A more detailed model may require larger amount of data until it is trained properly. Thus, poor prediction performance at early times does not necessarily imply poor prediction performance at later times.
%\item (\textit{Training patterns.}) Providing an informative answer to Question~\ref{Q:Question2}, requires a synchronous training of all available models any time new data arrives. However, such synchronous training imposes a high computational burden to addressing optimization (\ref{eq:InitialOptimizationProblem}). If a less intensive training pattern is considered instead (e.g., by training only a few randomly selected models), then such a pattern may influence the answer to Question~\ref{Q:Question2}.
%\end{enumerate}

\subsection{Motivation: thermal dynamics in buildings}		\label{sec:Motivation}

To motivate further the above questions, we will discuss them within the context of a specific application, namely the problem of \emph{temperature prediction in residential buildings} (cf.,~\cite{ChasparisNatschlaeger16}). As presented and analyzed in \cite{ChasparisNatschlaeger16}, the thermal-mass dynamics describing the indoor-temperature evolution of a residential building corresponds to a bilinear system, since the control-input (i.e., the flow of the thermal medium) appears in products with either state and disturbance variables of the system. It is rather common to assume a linear approximation of the model, and employ linear transfer models for system identification and prediction (see, e.g., the MIMO ARMAX model in \cite{YiuWang07}, the MISO ARMAX model for HVAC systems in \cite{Scotton13} and the ARX model structure considered in \cite{Malisani10}). Depending on the operating conditions, this might be a good approximation, however variations in the control and disturbance variables may lead to bad approximations. In such cases, we need to either retrain our model or replace it with a more detailed one (see, e.g., the nonlinear regression models developed in \cite{ChasparisNatschlaeger16}).

Assuming that a model predictive controller is used to regulate the indoor temperature, we should expect that indoor-temperature predictions are requested often (e.g., every one hour). Although the use of a detailed and accurate prediction model would be the most obvious choice,  larger training times may prohibit its use from the very beginning. On the other hand, simpler models (e.g., linear transfer models) may be a more appropriate choice, especially for the beginning of the implementation, or when the operating conditions change rather often. As expected, \emph{the best model $S_j^*$ may not be unique for all $j$, since better predictions may be provided by simpler models at least at earlier times}.

\begin{figure}[t!]
\centering
\iffigures
\includegraphics[scale=0.8]{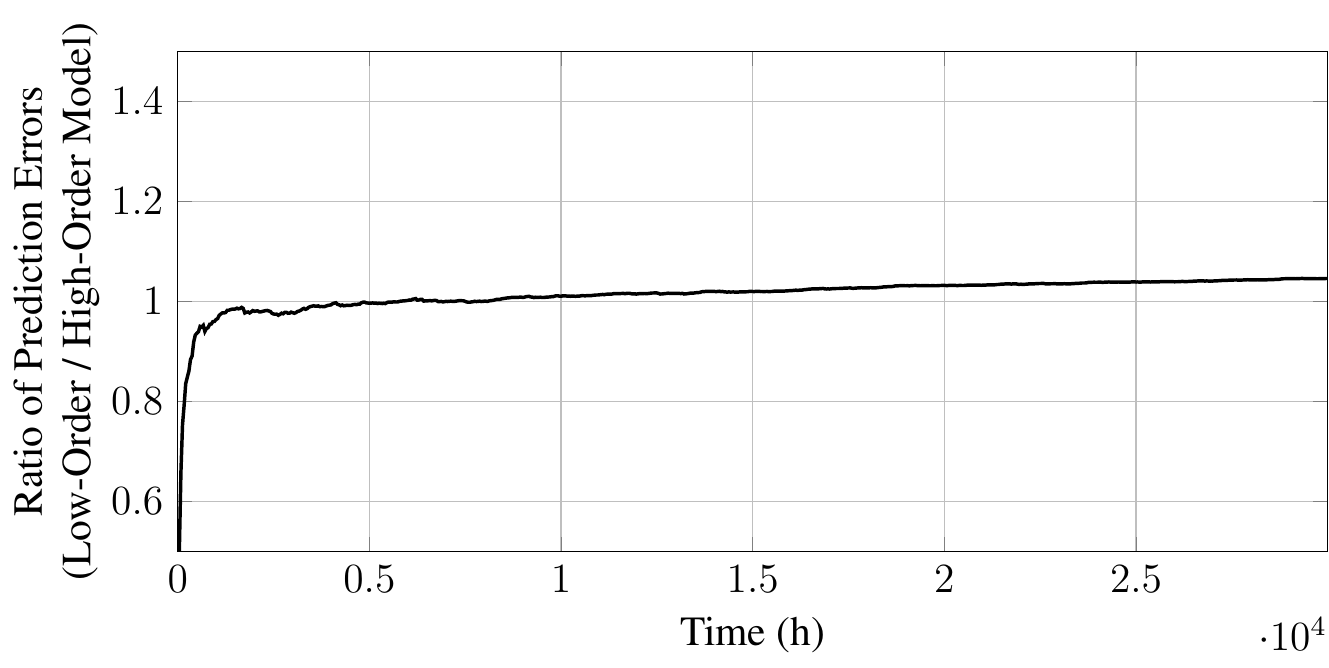}
\fi
\caption{Running average ratio of the squared prediction error between a low- and high-order model when used for indoor temperature prediction in a residential building.}
\label{fig:LinearModelsComparison}
\end{figure}

To demonstrate the validity of this argument (i.e., the fact that training times may be different among different models), Figure~\ref{fig:LinearModelsComparison} shows the ratio of the prediction error between a low- and high-order linear transfer model. The low-order model corresponds to a second-order output-error regression model, while the high-order model corresponds to a third-order output-error regression model (see, \cite[Section~5.4]{ChasparisNatschlaeger16} for the construction of such output-error models as well as the data collection process). Apparently, high-order models may exhibit low prediction accuracy at earlier times and high prediction accuracy at later times, due to the fact that more data need to be used for training such models (compared to lower-order models). 

From this example, we verify that the best model $S_j^*$ may change with time (either due to the different training times of the available models or due to varying operating conditions). Hence, we need a computationally efficient way that provides an answer to the optimization problem (\ref{eq:InitialOptimizationProblem}) that runs frequently and as soon as new data become available. Ideally, we would like to avoid evaluating and retraining \emph{all} available models each time new measurements become available, since this will impose a large computational burden to the proposed methodology.

\subsection{Objective}

\emph{\textbf{The goal of this paper is to provide an answer to Questions~\ref{Q:Question2}--\ref{Q:Question3}}}, i.e., we wish to provide a methodology that a) \emph{computes $\hat{S}_j$ in a computationally efficient way that is adaptive to variations in the performances of the models (Question~\ref{Q:Question3})}, and b) \emph{provides guarantees under which $\hat{S}_j\equiv S_j^*$ (Question~\ref{Q:Question2})}.

In particular, we will introduce a supervisory scheme that continuously switches the selection of the prediction model (based on \emph{prior} performances of the models in $\cM$), while it only trains a model when this is selected. In parallel, a partitioning of the inputs domain will also be considered that allows for selecting a different prediction model in each partition set, thus exploiting the bilinear form of the dynamics. The supervisory scheme for switching between models will also allow for switching between partitions, which creates more possibilities with respect to the optimal prediction (switched) model.

%%%%%%%%%%%%%%%%%%%%%%%%%%%%%%%%%%%%%%%%%%%%%%%%%%%%%%%%%%%%%%%%%%%%%%%%%%%%%%%%%%%%%%%%%%%%
\section{Switching Supervisory Online Prediction} \label{sec:SwitchingSupervisoryID}

%~~~~~~~~~~~~~~~~~~~~~~~~~~~~~~~~~~~~~~~~~~~~~~~~~~~~~~~~~~~~~~~~~~~~~~~
\subsection{Model selection}	\label{sec:Algorithm}

We assume that a set $\cM$ of prediction models is available for generating output predictions. Predictions are requested at specific instances of period $T_e=\rho T_s$ for some $\rho \in\mathbb{N}$. Let us denote these \emph{evaluation instances}, $0,T_e,2T_e,...$, by the index $k=0,1,2,...$, respectively. At each evaluation instance $k$, the switching supervisory scheme is responsible for selecting a prediction model $S_{k+1}\in\cM$ that will be used for generating predictions over the following $N$ \emph{sampling instances} $\{\tau_{kN+1},...,\tau_{(k+1)N}\}$. In general, and depending on the application, we should expect that $\rho\leq{N}$.

The selection of the model $S_k$, for each $k=0,1,2,...$, will be based upon a probability distribution that captures our estimates for the most ``\emph{appropriate}'' model. In particular, we define a probability vector $\vv_k \in \Delta(\magn{\cM})$. For example, in case of three available models in $\cM$, a vector of the form $\vv_k=(0.2,0.1,0.7)$ represents our estimates (or \emph{beliefs}) for the most appropriate model among the three available ones. An initial uniform distribution is assumed, i.e., $\vv_{0} \df (\nicefrac{1}{\magn{\cM}},...,\nicefrac{1}{\magn{\cM}})$.

The specifics of the supervisory algorithm are presented in Table~\ref{Tb:SOPA_ModelSelection}.

\begin{table}[t!]
\boxed{\small
\begin{minipage}{0.95\textwidth}
\vspace{5pt}
{\bf Algorithm 1.} \emph{Model Selection}\\[-2pt]
\noindent\rule{\textwidth}{0.4pt}\\[5pt]
At any evaluation instance $k=1,2,...$, the following steps are executed recursively:
\begin{enumerate}
\item (\emph{\textbf{collection}}) We update the history of sampled input-output pairs, i.e.,
$$H_{kN} \df H_{(k-1)N}\bigcup_{j=(k-1)N+1}^{kN} (\mathbf{u}(\tau_{j}),\mathbf{y}(\tau_j)).$$

\item (\emph{\textbf{evaluation}}) We compute the performance of the currently selected prediction model $S_k$ as follows:
\begin{eqnarray}
r_{k}(S_{k}) \df \left( \frac{1}{N} \sum_{j=(k-1)N}^{kN-1} \left\|\mathbf{y}(\tau_{j+1}) - \hat{\mathbf{y}}_k(\tau_{j+1})\right\|^2_2 \right)^{-1}
\end{eqnarray}
where $\hat{\mathbf{y}}_k(\tau_{j+1}) \df S_k(H_{j})$ is the prediction of the model $S_{k}$ given history $H_{j}$.

\item (\emph{\textbf{update}}) We revise the probability vector $\vv_{k}$:
\begin{equation}	\label{eq:ModelsStrategiesUpdates}
\vv_{k+1} = \vv_{k} + \epsilon \cdot r_{k}(S_{k}) \cdot \left[ \ve_{S_{k}} - \vv_{k} \right],
\end{equation}
for some constant $\epsilon>0$, where $\ve_{S_k}$ is the unit vector of size $\magn{\cM}$ such that the entry corresponding to model $S_k$ is 1 and all others are 0.

\item (\emph{\textbf{selection}}) We update the prediction model as follows:
\begin{eqnarray*}
S_{k+1} = \begin{cases}
	{\rm rand}_{\vv_{k+1}}[\cM], & {\rm w.p.} \quad 1-\lambda \\
    {\rm rand}_{\rm unif}[\cM], & {\rm w.p.} \quad \lambda	
	\end{cases}
\end{eqnarray*}
for some small perturbation factor $\lambda \in (0,1)$.

\item (\emph{\textbf{training}}) We train the parameters of the newly selected model, i.e., $$S_{k+1} \leftarrow \mathcal{T}(S_{k+1}|{H}_{kN}),$$ where $\mathcal{T}$ is a (user-defined) training operator (possibly different for each model in $\cM$). 

\end{enumerate}
\end{minipage}
}
\caption{Algorithm 1: Switching algorithm for model selection.}
\label{Tb:SOPA_ModelSelection}
\end{table}

Note that in Step 2 (\emph{evaluation}), the performance of $S_{k}$ is computed as the inverse average prediction error over the previous evaluation interval. In Step 3 (\emph{update}), given the performance of model $S_{k}$, we update the probability distribution $\vv_k$ that holds our estimates for the most appropriate model. For example, if $S_{k}$ corresponds to the $j$th entry in $\vv_k$, then $\vv_{k}$ is going to increase in the $j$th direction by $\epsilon \cdot r_{k}(S_{k}) \cdot (1-\vv_{k}[j])>0$ and decrease in any other direction $i\neq{j}$ by $\epsilon \cdot r_{k}(S_{k}) \cdot \vv_{k}[i]$. In other words, \emph{$\vv_k$ is reinforced in the direction of prior selections and proportionally to the performance of the selected models}. In Step 4 (\emph{selection}), a new model is selected, $S_{k+1}$, that will be used for the predictions over the next evaluation interval. The selection is based upon our updated estimate vector $\vv_{k+1}$. Finally, at the last step, the selected model $S_{k+1}$ will be trained using the updated history ${H}_{kN}$. In other words, only the model selected at Step~4 is trained at each evaluation instance $k$. 

The recursion (\ref{eq:ModelsStrategiesUpdates}) governs the asymptotic properties of the model selection and will be discussed in a forthcoming section as part of a more general scheme. This class of dynamics corresponds to a discrete-time analog of the replicator dynamics (cf.,~\cite[Chapter~7]{HofbauerSigmund98}), introduced in \cite{ChasparisShamma11_DGA,ChasparisShammaRantzer14_IJGT} under a game formulation. The difference here is that the performance received for a given selection of models may not be fixed with time since it depends on the history $H_{kN}$.

\subsection{Bilinearity and switched-model approximation}

The supervisory prediction scheme of Table~\ref{Tb:SOPA_ModelSelection} is independent of the specifics of the system dynamics. In this section, we wish to exploit the bilinear structure of the considered dynamics to improve the prediction performance. 

Note, that bilinear systems can well be approximated by linear systems under the assumption that the control input(s) are sufficiently constant. We introduce a partitioning of the domain $\mathcal{U}_i$ for each one of the input variables $u_i$. A finite set of partition patterns is considered for each input $u_i$, denoted by $\cP_i$. Let $p_i$ be an enumeration of the patterns in $\mathcal{P}_i$, i.e., $p_i=1,2,...,\magn{\mathcal{P}_i}$. A partition pattern may correspond to dividing the domain equally into two parts, three parts, etc., or any other (predefined) pattern. Given the selected partition pattern $p_i$ for each control input $u_i$, we formulate the \emph{partition-pattern profile} $\vp=(p_1,p_2,...,p_m)$ for all input variables. Figure~\ref{fig:PartitioningExample} demonstrates the outcome of such partition profile $\vp$ in case of two inputs $u_1$ and $u_2$ where the selected partition patterns $p_1$ and $p_2$ correspond to dividing the domain into two subsets.

The partition-pattern profile $\vp=(p_1,p_2,...,p_m)$ divides the Cartesian product $\mathcal{U}$ into a collection of subsets for the combined input variables. For example, in Figure~\ref{fig:PartitioningExample} such \emph{subset combinations} include $\mathcal{U}_{11}\times\mathcal{U}_{21}$, $\mathcal{U}_{11}\times\mathcal{U}_{22}$, $\mathcal{U}_{12}\times\mathcal{U}_{21}$, and $\mathcal{U}_{12}\times\mathcal{U}_{22}$. Let us denote the set of these \emph{subset combinations} by $\cA=\Ap$. To distinguish between the generated subset combinations, we introduce an enumeration $\alpha=\alpha(\vp)$ of the elements of $\cA$. To simplify notation, $\alpha$ will also denote the corresponding subset, thus $\mathbf{u}\in\alpha$ translates into the control input belonging to the subset corresponding to $\alpha$.

\begin{figure}[t!]
\centering
\iffigures
\begin{minipage}{0.45\textwidth}
\includegraphics[scale=0.8]{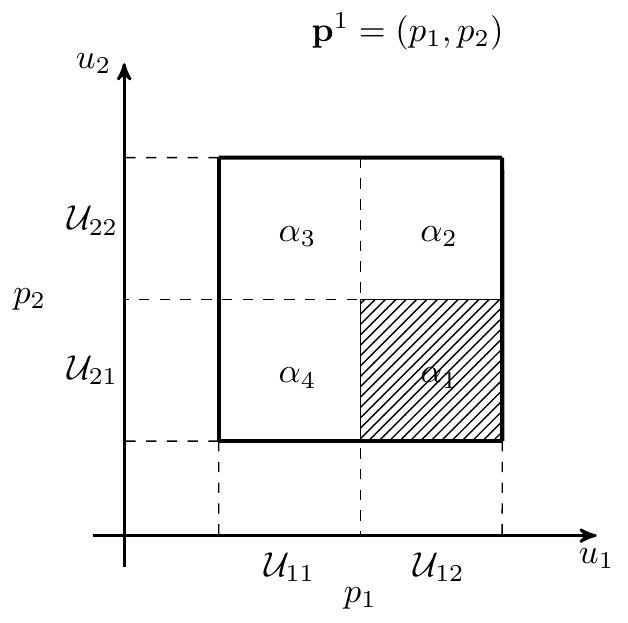}
\end{minipage}
\begin{minipage}{0.45\textwidth}
\includegraphics[scale=0.8]{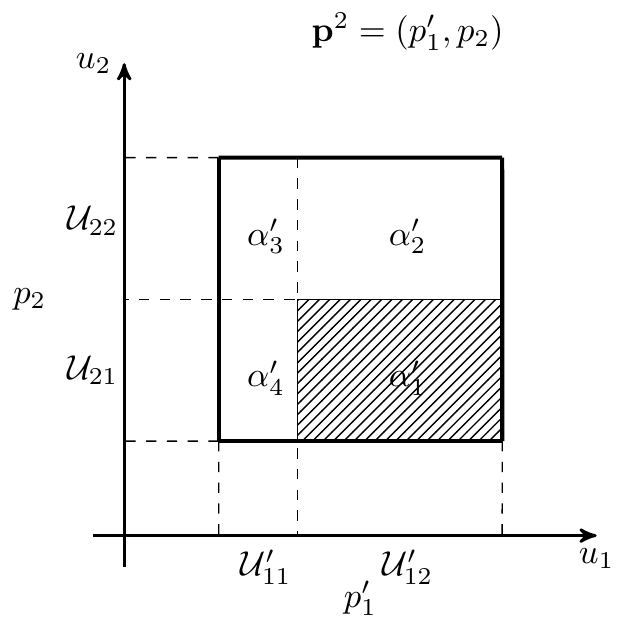}
\end{minipage}
\fi
\caption{Two partition-pattern profiles, $\vp^{1}$ and $\vp^2$, which differ in the way the domain $\mathcal{U}_1$ of control input $u_1$ is partitioned.}
\label{fig:PartitioningExample}
\end{figure}

Naturally, the enumeration $\alpha=\alpha(\vp)$ may be used to define a switched-system for the prediction of the output (\ref{eq:BilinearDynamics}). In particular, for each subset $\alpha$, we may assign a prediction model that only applies within this subset $\alpha$, leading to a switched prediction model. 

In other words, under a partition pattern $\vp$, and for each subset $\alpha\in\cA(\vp)$ generated from this partition, we assign a prediction model $S_{\vp\alpha,k}$ for each evaluation interval $k$. Then, the prediction at time $\tau_j$, for $j\in\{kN,...,(k+1)N-1\}$, can be expressed by 
\begin{equation}	\label{eq:SwitchedModel}
\hat{y}(\tau_{j+1}) = S_{\vp\alpha,k}(H_{j}),
\end{equation}
where $\alpha$ is such that $\mathbf{u}(\tau_j)\in\alpha$.

In this case, we may record the performance of a prediction model separately for each subset $\alpha$, defined as follows:
\begin{eqnarray}	\label{eq:UpdatedEvaluationFunction}
r_{\vpq,k}(S_{\vpq,k}) \df \left( \nicefrac{1}{\eta(k)} \sum_{j=(k-1)N}^{kN-1} \|\mathbf{y}(\tau_{j+1}) - \hat{\mathbf{y}}_k(\tau_{j+1})\|^2_2 \mathbb{I}_{\mathbf{u}(\tau_j)\in\alpha} \right)^{-1}
\end{eqnarray}
where $\eta(k) \df \sum_{j=(k-1)N}^{kN-1}\mathbb{I}_{\mathbf{u}(\tau_j)\in\alpha}$ and $\mathbb{I}_{\mathbf{u}(\tau_j)\in\alpha}$ is the index function. Accordingly, we may assign a separate strategy vector for each one of these subsets, denoted $\vv_{\vp\alpha}$, which can be updated (similarly to (\ref{eq:ModelsStrategiesUpdates}) as follows: 
\begin{equation}	\label{eq:ModelsStrategiesUpdatesB}
\vv_{\vpq,k+1} = \vv_{\vpq,k} + \epsilon \cdot r_{\vpq,k}(S_{\vpq,k}) \cdot \left[ \ve_{S_{\vpq,k}} - \vv_{\vpq,k} \right],
\end{equation}
for each $\vp$ and $\alpha\in\Ap$. 

Thus, the model selection process described by Algorithm~1 (Table~\ref{Tb:SOPA_ModelSelection}) can now be applied \emph{separately} for each one of the subsets in $\mathcal{A}(\vp)$.

\subsection{Partition selection}

\begin{table}[t!]
\boxed{\small
\begin{minipage}{0.95\textwidth}
\vspace{5pt}
{\bf Algorithm 2.} \emph{Partition Selection}\\[-2pt]
\noindent\rule{\textwidth}{0.4pt}\\[5pt]
At any evaluation instance $k=1,2,...$, the following steps are executed recursively:
\begin{enumerate}
\item (\emph{\textbf{collection}}) We update the history of sampled input-output pairs, i.e.,
$$H_{kN} \df H_{(k-1)N}\bigcup_{j=(k-1)N+1}^{kN} (\mathbf{u}(\tau_{j}),\mathbf{y}(\tau_j)).$$

\item (\emph{\textbf{evaluation}}) We compute the performance of the selected partition $\vp_k$,
\begin{equation}	\label{eq:PartitionPatternPerformance}
R_k(\vp_k,\vS_{k}) \df \sum_{\alpha\in\Ap} r_{\vpq,k}(S_{\vpq,k}),
\end{equation}
i.e., it aggregates the performances of all subsets $\alpha=\alpha(\vp)$ generated from $\vp$. We set it as a function of the model profile, $\vS_k\df\{\vS_{\vp,k}\}_{\vp\in\cP}$, i.e., the collection of models currently selected over all partition patterns. 

\item (\emph{\textbf{update}}) We revise the strategy vector over partition profiles $\vw_{k}$:

\begin{eqnarray}	\label{eq:PartitionsStrategyUpdates}
\vw_{k+1} = \vw_k + \epsilon \cdot R_{k}(\vp_k,\vS_{k}) \cdot [\ve_{\vp_k}-\vw_k],
\end{eqnarray}
with $\vw(0)\df(\nicefrac{1}{\magn{\cP}},...,\nicefrac{1}{\magn{\cP}})$.

\item[3.] (\emph{\textbf{selection}}) We update the partitions as follows:
%\begin{eqnarray*}
%S_{k+1} = \begin{cases}
%	{\rm rand}_{\vv_{k+1}}[\cM], & {\rm w.p.} \quad 1-\lambda \\
%    {\rm rand}_{\rm unif}[\cM], & {\rm w.p.} \quad \lambda	
%	\end{cases}
%\end{eqnarray*}
\begin{eqnarray*}	
  \vp_{k+1} = \begin{cases}
    {\rm rand}_{\vw_{k+1}}[\cP], & \mbox{w.p.} \quad 1-\lambda, \\
    {\rm rand}_{\rm unif}[\cP], & \mbox{w.p.} \quad \lambda,
  \end{cases}
\end{eqnarray*}
for some small perturbation factor $\lambda \in (0,1)$.

%\item (\emph{\textbf{training}}) We train the parameters of the newly selected model, i.e., $$S_{k+1} \leftarrow \mathcal{T}(S_{k+1}|{H}_{kN}),$$ where $\mathcal{T}$ is a (user-defined) training operator (possibly different for each model in $\cM$). 

\end{enumerate}
\end{minipage}
}
\caption{Algorithm 2: Switching algorithm for partition selection.}
\label{Tb:SOPA_PartitionSelection}
\end{table}

Given that the partition-pattern profile $\vp$ may not be unique, we also wish to provide an answer to \emph{which partition-pattern may be more appropriate for the switched-model prediction}. Providing a \emph{predefined} finite set of partition-pattern profiles $\cP=\cP_1\times\cdots\times\cP_m$, we also implement an update of the partition-pattern profile $\vp_k\in\cP$ at each evaluation instance $k$, as described by Algorithm~2 in Table~\ref{Tb:SOPA_PartitionSelection}. 

In particular, we introduce a probability (or \emph{strategy}) vector $\vw_k\in\Delta(\magn{\cP})$ that holds our estimates for the most appropriate partition-pattern profile $\vp\in\cP$ (similarly to the use of $\vv_{\vpq,k}$ for the most appropriate model in $\alpha$). Then, at every evaluation instance $k$, we record the overall performance of the partition, as this is defined in (\ref{eq:PartitionPatternPerformance}), and we update the strategy vector over partitions according to (\ref{eq:PartitionsStrategyUpdates}). 

\subsection{Joint model \& partition selection}

Note that either Algorithm~1 for model selection or Algorithm~2 for partition selection can be implemented independently or jointly. In particular, Algorithm~1 may be implemented in an unpartitioned domain, or it may be applied for each subset $\alpha$ of a partitioned domain $\vp$. Similarly, Algorithm~2 may be applied assuming that the prediction model in each subset is fixed, or it may be applied jointly with the model selection process.

In the remainder of this paper, we will consider the most general case, which is a joint execution of Algorithm~1 and Algorithm~2 at each evaluation interval $k=1,2,...$. The joint execution is described in detail in Table~\ref{Tb:SOPA_JointSelection}. It is important to note that the joint execution is governed by Algorithm~2 which determines the currently selected partition-pattern. Then, for the currently selected pattern (and only for this one), Algorithm~1 is also implemented in \emph{each} one of its subsets.

\begin{table}[t!]
\boxed{\small
\begin{minipage}{0.95\textwidth}
\vspace{5pt}
{\bf Algorithm 3.} \emph{Joint Partition-Model Selection}\\[-2pt]
\noindent\rule{\textwidth}{0.4pt}\\[5pt]
At any evaluation instance $k=1,2,...$, let $\vp_k$ be the currently selected pattern, and $\vS_{\vp_k,k}$ be the currently selected models for the subsets of the selected profile.
\begin{enumerate}
\item[(a)] Run Algorithm~2, i.e., \emph{evaluate} pattern $\vp_k$, \emph{update} the corresponding strategy vector $\vw_k$ over patterns, and \emph{select} a new pattern, $\vp_{k+1}$, which will be applied in the new evaluation interval, $k+1$.

\item[(b)] Run Algorithm~1 \emph{for each subset $\alpha\in\mathcal{A}(\vp_k)$ of the selected pattern}, i.e., for each subset $\alpha\in\mathcal{A}(\vp_k)$, \emph{evaluate} the selected model $S_{\vp_k\alpha,k}$ and \emph{update} the strategy vector $\vv_{\vp_k\alpha,k}$ over models. Finally, \emph{select} new models $\vS_{\vp_{k+1},k+1}$, for the newly selected pattern $\vp_{k+1}$ (which may \emph{not} coincide with the previous one $\vp_k$). 

\end{enumerate}
\end{minipage}
}
\caption{Algorithm 3: Switching algorithm for joint partition-model selection.}
\label{Tb:SOPA_JointSelection}
\end{table}

\subsection{Illustration}	\label{sec:Illustration}

\subsubsection{Setup}

Let us consider a simple example, where $\mathcal{M}=\{S^1,S^2\}$ is the set of two prediction models, and $\mathcal{P}=\{\mathbf{p}^1,\mathbf{p}^2\}$ are two partition-pattern profiles of the input domain, (see, e.g., Figure~\ref{fig:PartitioningExample}). Both sets $\cM$ and $\cP$ have been selected/constructed by the user.

The user formulates a vector of strategies/probabilities over the available partition profiles, 
\begin{eqnarray*}
\vw_0 = \left(\begin{array}{c}
\mbox{\rm prob. of selecting partition $\vp^1$}\\
\mbox{\rm prob. of selecting partition $\vp^2$}
\end{array}\right) = \left(\begin{array}{c}
\nicefrac{1}{2}\\
\nicefrac{1}{2}
\end{array}\right).
\end{eqnarray*}
At the beginning, we may assume an equal probability assigned to each partition, since we do not have any a-priori knowledge regarding which might be the best selection.

Each partition-pattern $\vp^1$ and $\vp^2$, has divided the control-input domain into a collection of subsets, indicated by the enumerations $\alpha=\alpha(\vp^1)\in\{\alpha_1,\alpha_2,\alpha_3,\alpha_4\}$ and $\alpha'=\alpha'(\vp^2)\in\{\alpha_1',\alpha_2',\alpha_3',\alpha_4'\}$, respectively. Each of these indices, corresponds to a different subset in the control domain as shown in Figure~\ref{fig:PartitioningExample}.

For each partition-pattern profile, say $\vp^1$, and for each subset generated from that, $\alpha$, we create a vector of strategies, $\vv_{\vp^1\alpha}$, representing our beliefs of what would be the most appropriate model for this subset. Since two models are available $\cM = \{S^1,S^2\}$, we set
\begin{eqnarray*}
\vv_{\vp^1\alpha,0} = \left(\begin{array}{c}
\mbox{\rm prob. of selecting system $S^1$}\\
\mbox{\rm prob. of selecting system $S^2$}
\end{array}
\right) = \left(\begin{array}{c}
\nicefrac{1}{2}\\
\nicefrac{1}{2}
\end{array}\right), \ \ \alpha\in\cA(\vp)\equiv\{\alpha_1,\alpha_2,\alpha_3,\alpha_4\}.
\end{eqnarray*}
We do the same for all the partition-patterns, thus each subset has one such vector assigned to it.

Having initialized all vectors $\vv_{\vp^1\alpha}$ and $\vv_{\vp^2\alpha'}$, we make an initial assignment, that is we assign a prediction model to each one of the subsets. Let us assume that we assign the first model, i.e., $S_{\vp^1\alpha,0} = S^1$ and $S_{\vp^2\alpha',0}=S^1$ for each subset $\alpha$, $\alpha'$. Let us also assume that we make an initial partition selection, $\vp_{0} = \vp^1$, i.e., we select the first partition of Figure~\ref{fig:PartitioningExample}.

\subsubsection{Evaluation \& Update}

At the next evaluation instance, $k=1$, we apply the following updates:
\begin{enumerate}
\item Algorithm 1 (Table~\ref{Tb:SOPA_ModelSelection}) for evaluating and updating the models of each subset of the \emph{selected} partition ($\vp_0=\vp^1$);
\item Algorithm 2 (Table~\ref{Tb:SOPA_PartitionSelection}) for evaluating and updating the partition selection.
\end{enumerate}
In particular, we go through each subset $\alpha$ of the selected partition $\vp^1$, and record the performance of the selected model, i.e., $r_{\vp^1\alpha,0}(S^1)$, $\alpha\in\{\alpha_1,\alpha_2,\alpha_3,\alpha_4\}$. Then, we update the strategies in each subset separately, i.e.,
\begin{eqnarray*}
\vv_{\vp^1\alpha,1} = \vv_{\vp^1\alpha,0} + \epsilon \cdot r_{\vp^1\alpha,0}(S^1) \cdot \left[\left(\begin{array}{c}1\\0\end{array}\right) - \vv_{\vp^1\alpha,0}\right] = \nicefrac{1}{2}\left(\begin{array}{c}
1 + \epsilon\cdot r_{\vp^1\alpha,0}(S^1) \\
1 - \epsilon\cdot r_{\vp^1\alpha,0}(S^1)
\end{array}\right), 
\end{eqnarray*}
Note that the strategy in each subset $\alpha$ of the selected partition $\vp^1$ will increase in the direction of system $S^1$, since this was the system selected at the previous iteration. This increase is proportional to the performance observed. Correspondingly, the probability of selecting system $S^2$, will decrease by the same quantity.

After updating the strategies over models of each subset of the selected partition $\vp^1$, we apply Algorithm~2 for evaluating/updating the partition. In particular, we first compute the aggregate performance of the partition 
\begin{equation*}
R_0(\vp_0,\vS_0) = R_0(\vp^1,\{S^1,S^1,S^1,S^1\}) = \sum_{\alpha\in\{\alpha_1,\alpha_2,\alpha_3,\alpha_4\}}r_{\vp^1\alpha,0}(S^1).
\end{equation*}
Given the recorded performance of this partition, we update the strategy vector over partitions. The updated strategy vector will be:
\begin{eqnarray*}
\vw_{1} = \vw_{0} + \epsilon \cdot R_{0}(\vp_0,\vS_0) \cdot \left[\left(\begin{array}{c}1\\0\end{array}\right) - \vw_{0}\right] = \nicefrac{1}{2}\left(\begin{array}{c}
1 + \epsilon\cdot R_{0}(\vp_0,\vS_0) \\
1 - \epsilon\cdot R_{0}(\vp_0,\vS_0)
\end{array}\right).
\end{eqnarray*}
Having computed $\vw_1$, we make a new selection for the partition (as indicated by Step~4 of Algorithm~2), say $\vp^2$, which will be applied in the next evaluation interval, and for the selected partition, we make a selection over models (for each subset) using the strategy vectors $\vv_{\vp^2\alpha,1}$. We continue likewise.

Note that a strategy vector over models is only being updated when the corresponding partition-pattern profile is being selected.

\subsection{Discussion}

The switching-supervisory scheme of Algorithm~3 (Table~\ref{Tb:SOPA_JointSelection}) provides one possibility for automating the optimal switching strategy for partition-patterns and models. The following questions naturally emerge: a) \emph{Why this type of learning dynamics is appropriate for this problem?}, b) \emph{What is the memory/computational complexity of the learning dynamics?}, and c) \emph{What are the convergence guarantees of the proposed dynamics?}

\textbf{\textit{Performance-based optimization:}} One of the most attractive features of the proposed dynamics is the fact that decisions are taken based on measurements of performance functions. Thus, no a-priori knowledge of the performance of the models/partitions is required. This property is quite appropriate for the problem of output prediction, since the performance of a model may vary significantly with time (due to varying operating conditions), as already discussed in Section~\ref{sec:Motivation}.

\textbf{\textit{Memory complexity:}} For predefined sets $\cM$ and $\cP$, the joint partition-model update of Algorithm~3 requires that we keep track of: a) $\sum_{\vp\in\cP}\magn{\mathcal{A}(\vp)}$ strategy vectors (over models) of size $\magn{\cM}$ each, and b) a single strategy vector (over partition-pattern profiles) of size $\magn{\cP}$. Thus, the requested memory size is dominated by a number that is proportional to $(c\magn{\cM}+1)\magn{\cP}$, where $c$ reflects the maximum number of subsets of a partition pattern within $\cP$. Such feature is highly attractive since no prior performances need to be preserved throughout the implementation, except for the fixed memory size of the strategy vectors. 

Note that the above memory size computation does not include the history $H_k$ of input-output pairs (since this will be required by any scheme). However, if we allow models in $\cM$ that can be trained recursively, then only the most recent input-output pairs will be required at each iteration $k$.

\textit{\textbf{Computational complexity:}} According to the joint partition/model update, at each iteration $k$, the partition strategy vector is updated according to (\ref{eq:PartitionsStrategyUpdates}). Moreover, the model strategy vectors of the selected partition are updated according to (\ref{eq:ModelsStrategiesUpdates}). Updating a partition-pattern strategy vector requires a number of basic operations that is proporional to $\magn{\cP}$, while updating a model strategy vector requires a number of basic operations that is proportional to $\magn{\cM}$. Thus, the computational complexity of the requested computations will be dominated by a number that is proportional to $c\magn{\cM}+\magn{\cP}$, where $c$ is the maximum number of subsets of a partition pattern within $\cP$. In other words, the proposed scheme admits \emph{linear} computational complexity.

The resulting minimal memory requirements and the linear computational complexity of the proposed scheme provide an indirect answer to Question~\ref{Q:Question3}.

\textit{\textbf{Convergence guarantees:}} In the forthcoming Section~\ref{sec:Convergence}, we will show that this model/partition selection process will guarantee optimal prediction performance \emph{in probability}, that is the algorithm will provide the optimal selection for an arbitrarily large portion of the implementation time. We will also provide a condition under which Question~\ref{Q:Question2} is also answered. 

In conclusion, the intention here is to balance between a) low computational complexity of the dynamics (which is an important feature for supervisory schemes that run over long time horizons), and b) convergence guarantees to the optimal choices. Further discussion on the benefits of this type of dynamics will be provided at the end of the forthcoming Section~\ref{sec:Convergence}.

%%%%%%%%%%%%%%%%%%%%%%%%%%%%%%%%%%%%%%%%%%%%%%%%%%%%%%%%%%%%%%%%%%%%%%%%%%
\section{Convergence Analysis}	\label{sec:Convergence}

\subsection{Assumptions} \label{sec:Assumptions}

The evolution of the partition selection $\vp_k$ and the model selection $\vS_k$ is fully described by the vectors $\vw_k$ and $\{\vv_{\vpq,k}\}$, for each $\alpha(\vp_k)$, updated through recursions (\ref{eq:PartitionsStrategyUpdates}) and (\ref{eq:ModelsStrategiesUpdatesB}), respectively. \textit{\textbf{For the remainder of the paper}}, denote briefly $\vv$ to be the collection of vectors $\{\vv_{\vpq}\}_{\vp,\alpha}$. Define a canonical path space $\Omega$ with an element $\omega\in\Omega$ being a sequence $\{\vv_k,\vw_k,H_{kN}\}_k$ generated by the process. Define also $\mathfrak{F}_n$ to be a $\sigma$-algebra generated by the sequence up to time index $k\leq{n}$. Lastly, let $\mathbb{P}[\cdot]$ be a probability measure defined in $\mathfrak{F}_n$ and $\mathbb{E}_{\mathfrak{F}_n}[\cdot]$ be the expectation conditioned on $\mathfrak{F}_n$. 

% Let $\Omega$ denote the canonical path space with an element $\omega\in\Omega$ being a sequence $\{\vv_k,\vw_k,\mathcal{H}_{[\tau_{k-1},\tau_k)}\}_k$ generated by the process. We may define a random variable $\psi_k:\Omega\mapsto\cP\times\cM$, such that $\psi_k(\omega)=(\vp_k,\mathbf{S}_k)$ which corresponds to a selection of partition and model profile. Consider also the $\sigma$-algebra $\mathfrak{F}_n$ generated by the sequence up to time index $k\leq{n}$, a probability measure $\mathbb{P}[\cdot]$ defined in $\mathfrak{F}_n$ and let $\mathbb{E}_{\mathfrak{F}_n}[\cdot]$ denote the expectation conditioned on $\mathfrak{F}_n$. 

The convergence behavior is subject to the following assumptions. First, we assume that \emph{the prediction performance of the models in $\cM$ converges in the mean as time grows}. Formally, let $n_{\epsilon}$ be a sequence of positive integers, such that $n_{\epsilon}\to\infty$ and $\epsilon n_{\epsilon}\to{0}$ as $\epsilon\to{0}$.
\begin{assumption}[Mean asymptotic performance]	\label{As:MeanAsymptoticPerformance}
For each partition-pattern profile $\vp\in\cP$ and each subset $\alpha\in\Ap$, there exists a function $r_{\vpq,\infty}:\cM\mapsto\mathbb{R}_+$ such that, for any sequence $n_{\epsilon}$ and for any $S\in\cM$, the following holds:
\begin{equation*}	%\label{eq:MeanAsymptoticPerformanceCondition1}
\lim_{j\to\infty}\lim_{\epsilon\to{0}}\frac{1}{n_{\epsilon}}\sum_{k=jn_{\epsilon}}^{(j+1)n_{\epsilon}-1}\mathbb{E}_{\mathfrak{F}_{jn_{\epsilon}}}\left[r_{\vpq,k}(S) - r_{\vpq,\infty}(S)\right]=0.
\end{equation*}
\end{assumption}
In other words, let us consider a window of implementation time of the form $\{\tau_{jn_\epsilon},...,\tau_{(j+1)n_\epsilon-1}\}$ which consists of $n_\epsilon$ consecutive evaluation instances with starting instance being $jn_\epsilon$. Note that if we decrease $\epsilon$, the size of the window $n_\epsilon$ will increase, but its starting point will increase as well. Assumption~\ref{As:MeanAsymptoticPerformance}  states that, if we increase both the considered window of time and its starting point, then the average performance of a model $S\in\cM$ (in a partition $\vp$ and subset $\alpha$) will converge in the mean to some finite value. 

Given that every partition $\vp$ and model profile $\mathbf{S}$ will be picked infinitely often,\footnote{This observation follows directly from Borel-Cantelli lemma (cf.,~\cite[Lemma~3.14]{Breiman92}) and the fact that $\lambda>0$.} we should expect that the performance of a model on a subset $\alpha\in\Ap$ should approach a fixed value on average when the training operator is convergent. Variations in the performance of a model over a given subset $\alpha\in\Ap$ may always exist due to measurement noise or varying operating conditions, however the above assumption simply states that these variations should die away in the mean when the model has been trained sufficiently enough.
Assumption~\ref{As:MeanAsymptoticPerformance} constitutes a weaker condition compared to the conditions considered in an earlier work \cite{ChasparisNatschlaeger14_MSC}, since it accounts for noise perturbations in the performance of a model. 

\begin{assumption}[Distinct performances]	\label{As:DistinctPerformances}
For any partition-pattern profile $\vp\in\cP$, subset $\alpha\in\Ap$, and any two models $S,S'\in\cM$, we have $r_{\vpq,\infty}(S) \neq r_{\vpq,\infty}(S').$
\end{assumption}
Assumption~\ref{As:DistinctPerformances} simply states that any two distinct models in $\cM$ will provide distinct performances in the mean when trained sufficiently often, thus excluding trivial cases where the same prediction model is introduced more than once.

For any partition-pattern profile $\vp$ and model profile $\mathbf{S}$, introduce the following reward function:
$R_{\infty}(\vp,\mathbf{S}) \df \sum_{\alpha\in\Ap}r_{\vpq,\infty}(\Spq),$
which simply sums up the asymptotic average performances over all $\alpha\in\Ap$. Throughout the paper, the following assumption will also be considered.

\begin{assumption}[Best selection] \label{As:BestPartitionModels}
There exist a partition-pattern profile $\vp^*\in\cP$ and a model profile $\mathbf{S}^*$ such that:
\begin{enumerate}
 \item[(H1)] $r_{\vpq,\infty}(\Spq^*) > r_{\vpq,\infty}(\Spq')$, for all $\Spq' \neq \Spq^*$, $\vp\in\cP$ and $\alpha\in\Ap$;
 \item[(H2)] $R_{\infty}(\vp^*,\mathbf{S}^*) > R_{\infty}(\vp',\mathbf{S}^*)$, for all $\vp'\neq \vp^*$.
\end{enumerate}
\end{assumption}

This assumption implies that (H1) for a given partition-pattern profile $\vp$ and subset $\alpha\in\Ap$, there exists a model $\Spq^*\in\cM$ which performs better (in the mean) compared to any other model in $\cM$, and (H2) there exists a profile $\vp^*\in\cP$ which performs better than any other profile $\vp$ under $\mathbf{S}^*$. This assumption is always satisfied due to Assumption~\ref{As:DistinctPerformances}, since, there is always going to be a partition and a model profile which outperforms all the rest in the mean. Note further, that it is not required that the same model in $\cM$ is the best model for all pairs of $\vp\in\cP$ and $\alpha\in\Ap$, instead there exists one such model for each pair $\vp\in\cP$, $\alpha\in\Ap$.

\subsection{Convergence}

The asymptotic behavior of the supervisory process can be described by the probability distributions of $(\vv_k,\vw_k)$. In fact, convergence can be characterized with respect to the set of \emph{pure strategy} pairs $(\vv^*,\mathbf{w}^*)$ such that $\vv_{\vpq}^*$, $\mathbf{p}\in\cP$, $\alpha\in\Ap$ and $\vw^*$ are \emph{unit vectors}, (i.e., $\vw^*$ assigns probability one to a single partition-pattern profile in $\cP$ and each $\vv^*_{\vpq}$ assigns probability one to a single model in $\cM$).
\begin{theorem}[Convergence] \label{Th:GlobalConvergence}
Consider a step size $\epsilon>0$. Let Assumptions~\ref{As:MeanAsymptoticPerformance}, \ref{As:DistinctPerformances}, \ref{As:BestPartitionModels} hold and let $(\vv^*,\vw^*)$ be the pure strategy pair corresponding to the best selection $\vp^*$ and $\vS^*$ of Assumption~\ref{As:BestPartitionModels}. For sufficiently small $\lambda>0$, there exists $\delta=\delta(\lambda)>0$, with $\delta(\lambda)\to{0}$ as $\lambda\to{0}$, such that the fraction of time that the recursion $(\vv_k,\vw_k)$ spends in $\mathcal{N}_{\delta}(\vv^*,\vw^*)$ goes to one (in probability) as $\epsilon\to{0}$ and $k\to\infty$.
\end{theorem}
\begin{proof}
The proof requires a series of propositions and it will be shown in detail in Section~\ref{sec:TechnicalDerivation}.
\end{proof}

Theorem~\ref{Th:GlobalConvergence} establishes convergence (in probability) of $(\vv_k,\vw_k)$ into $(\vv^*,\vw^*)$ such that $\vv_{\vpq}^*$ and $\vw^*$ are unit vectors assigning probability 1 to the index of the best model $S_{\vpq}^*$ and the best partition $\vp^*$, respectively. \emph{Convergence in probability implies that the fraction of time that the recursion spends in a $\delta$-neighborhood of $(\vv^*,\vw^*)$ goes to one as $\epsilon\to{0}$ and $k\to\infty$}. In other words, as we reduce $\epsilon$, the fraction of time selecting the best partition/model profile becomes arbitrarily close to 1. If we also take $\lambda$ small, we may make the size of the $\delta$-neighborhood arbitrarily small thus enforcing selection of the best partition/model profile even further. 

Theorem~\ref{Th:GlobalConvergence} has been derived under Assumptions~\ref{As:MeanAsymptoticPerformance}--\ref{As:BestPartitionModels}. Note, however, that the mean performance of a model may not be constant with time, e.g., due to unseen operating conditions. Furthermore, the partition/model profile that provides the best performance may not necessary be fixed with time as well. The above convergence property of Theorem~\ref{Th:GlobalConvergence} applies as long as Assumptions~\ref{As:MeanAsymptoticPerformance}--\ref{As:BestPartitionModels} are valid. In case that, the mean performances alter significantly or the best partition/model profile changes, then the conclusions of Theorem~\ref{Th:GlobalConvergence} continue to hold, but for the new conditions. Thus, we may say that the algorithm exhibits adaptivity to possible changes in the performance of the models. This is attributed to the selection of constant step-size $\epsilon>0$.

% This adaptive response of the overall scheme is attributed to the selection of a constant step size (contrary to diminishing step size sequences) and also to the positive perturbation factor $\lambda$. If, instead, $\lambda=0$, then the above property will no longer be true, since the recursion would have converged to a pair of partition/model profile in finite time and it would have been unable to escape from it (as it is discussed in detail in \cite{ChasparisShammaRantzer14_IJGT}).

%%%%%%%%%%%%%%%%%%%%%%%%%%%%%%%%%%%%%%%%%%%%%%%%%%%%%%%%%%%%%%%%%%%%%%%%%%%%%%%%%%%%
\subsection{Evaluation}	\label{sec:Evaluation}

To evaluate the performance of the supervisory process, we considered the case-study of thermal dynamics in buildings presented in Section~\ref{sec:Motivation}. We then ran the supervisory prediction scheme of Table~\ref{Tb:SOPA_JointSelection}, for the prediction of the indoor temperature within one thermal zone. As described in \cite{ChasparisNatschlaeger16}, the control variables are a) the water flow rate $\dot{V}_{w}$ of the radiant-heating system, and b) the air flow rate $\dot{V}_{a}$ of the ventilation system. Using standard modeling techniques, reference \cite{ChasparisNatschlaeger16} has demonstrated that the nonlinear thermal dynamics can be written in the form of (\ref{eq:BilinearDynamics}).

We considered two partition patterns. The first corresponds to the non-partitioned set of flows, while the second divides the set of flows into two equal subsets (similarly to Figure~\ref{fig:PartitioningExample}). We designed a set of prediction models $\cM$, based on an output-error (discrete-time) formulation of the dynamics (cf.,~\cite[Section~4.2]{Ljung99}) which leads to predictors of the form $\varphi_{S}(\tau_j,\theta_{S})\tr\theta_S,$ $S\in\cM$. The unknown parameters $\theta_S$ are to be identified for each $S$. Three alternative regression vectors were considered, $\cM=\{1,2,3\}$, including a standard 2nd-order linear regression basis in model 1, a standard 2nd-order linear regression basis in model 2, and a 2nd-order nonlinear regression basis in model 3. 

\begin{figure}[t!]
\centering
\iffigures
\includegraphics[scale=0.8]{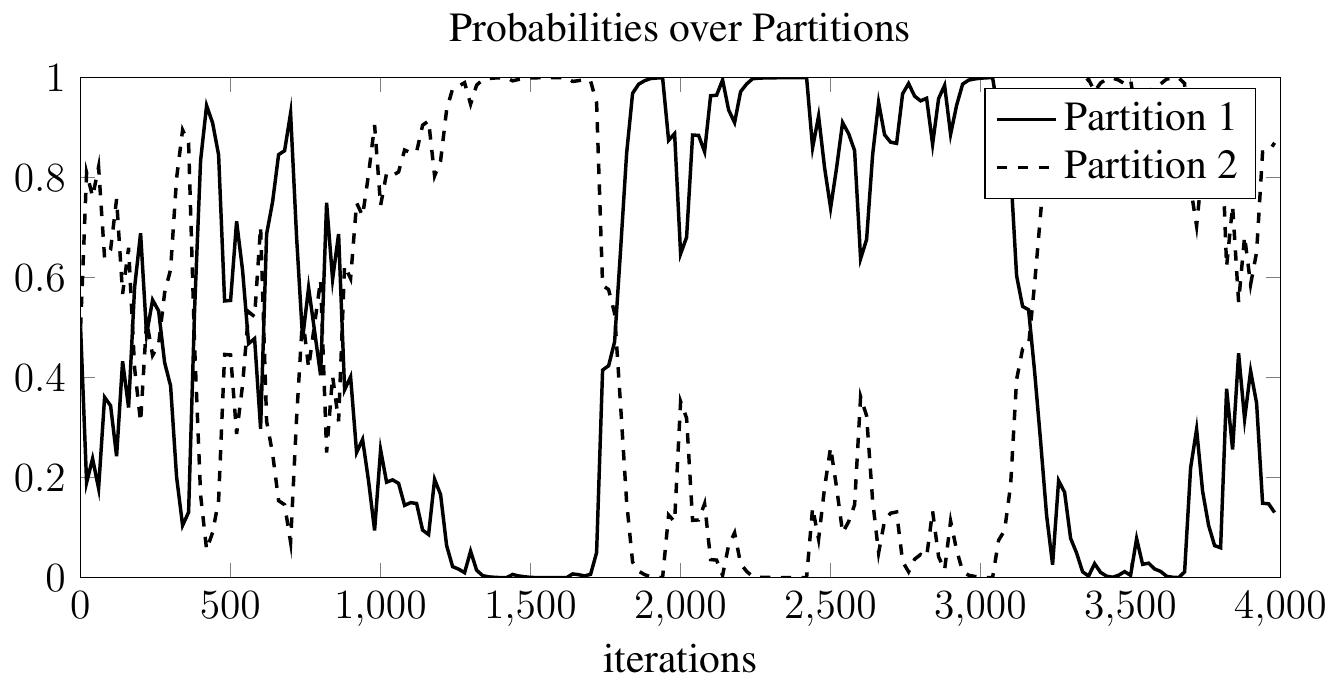}
\fi
\caption{Probabilities over partition-pattern profiles.}
\label{fig:ProbabilitiesOverPartitions}
\end{figure}

\begin{figure}[t!]
\centering
\iffigures
\includegraphics[scale=0.8]{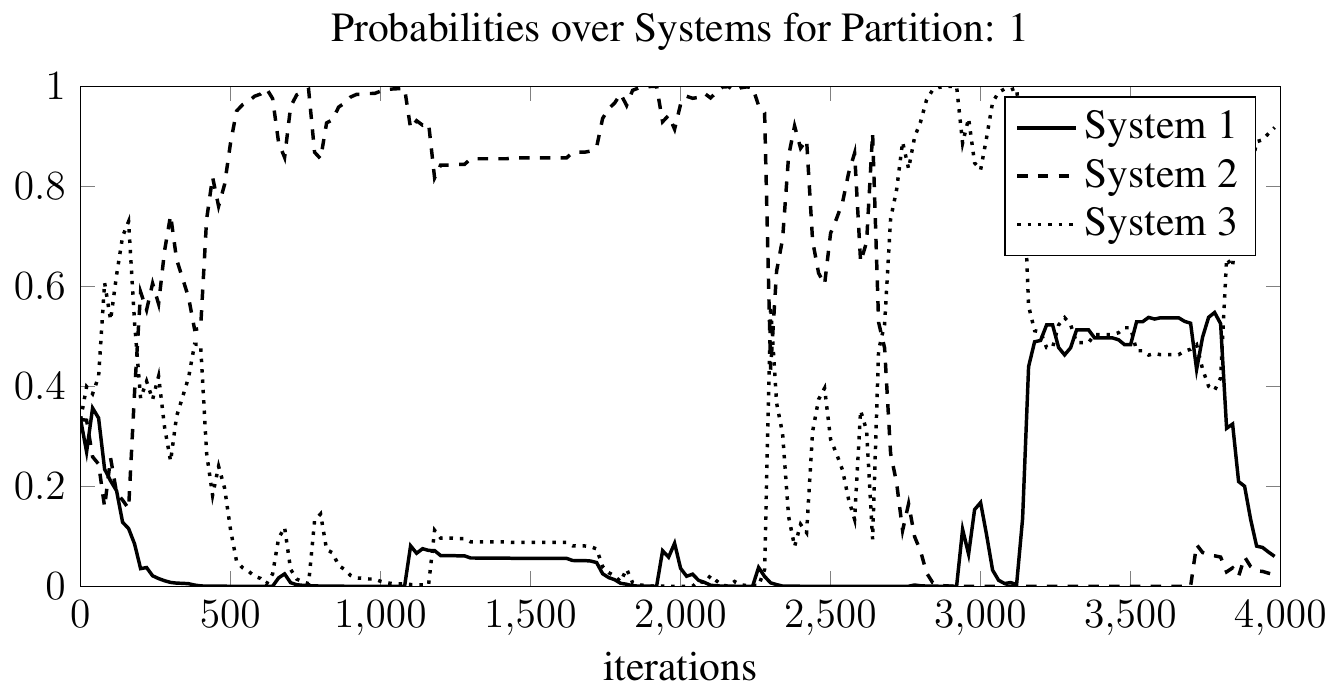}
\fi
\caption{Probabilities over models for partition scheme 1.}
\label{fig:ProbabilitiesOverModelsForPartition1}
\end{figure}

\begin{figure}[t!]
\centering
\iffigures
\includegraphics[scale=0.8]{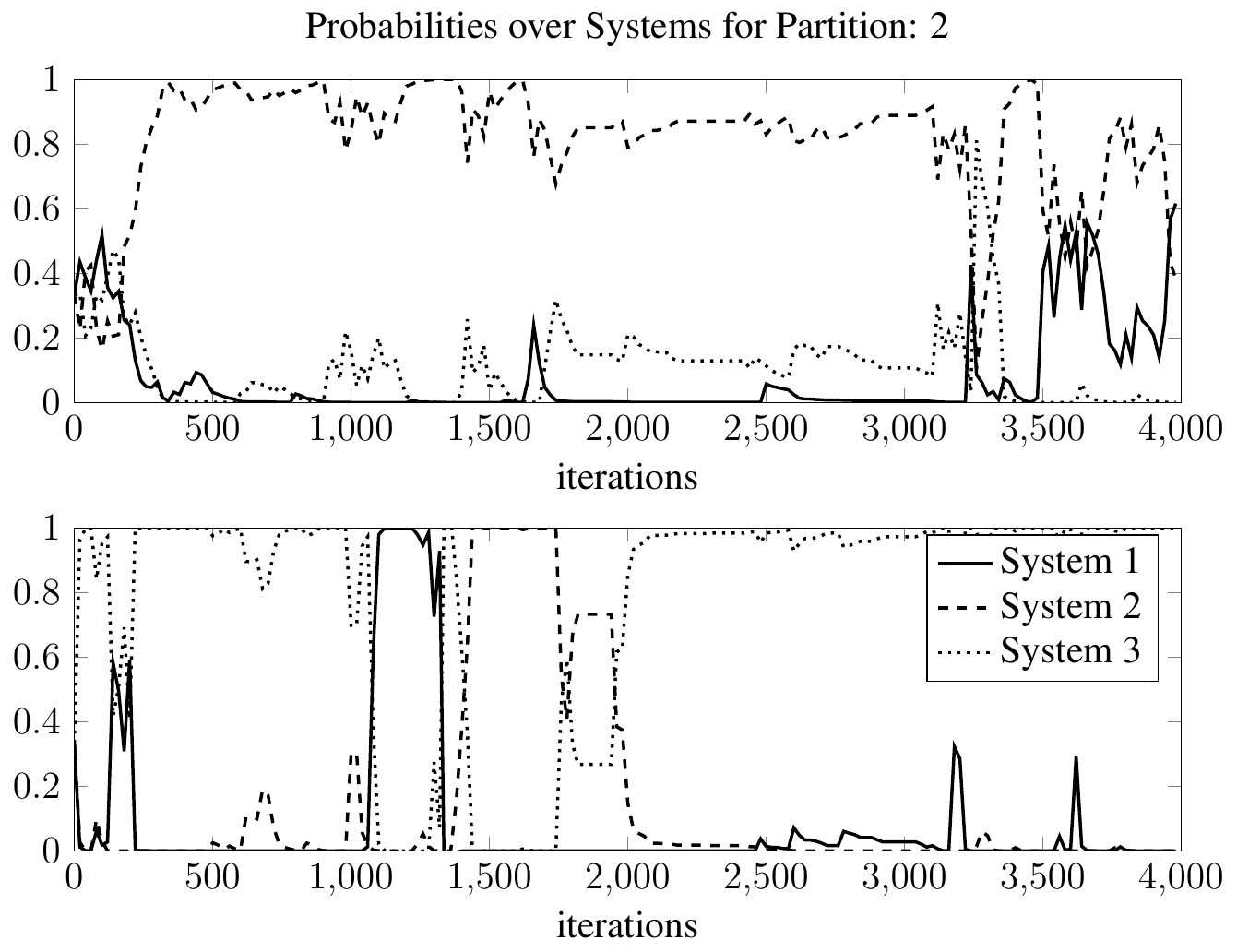}
\fi
\caption{Probabilities over models for partition scheme 2 and partition sets 1 \& 2, respectively.}
\label{fig:ProbabilitiesOverModelsForPartition2}
\end{figure}

\begin{figure}[t!]
\centering
\iffigures
\includegraphics[scale=0.8]{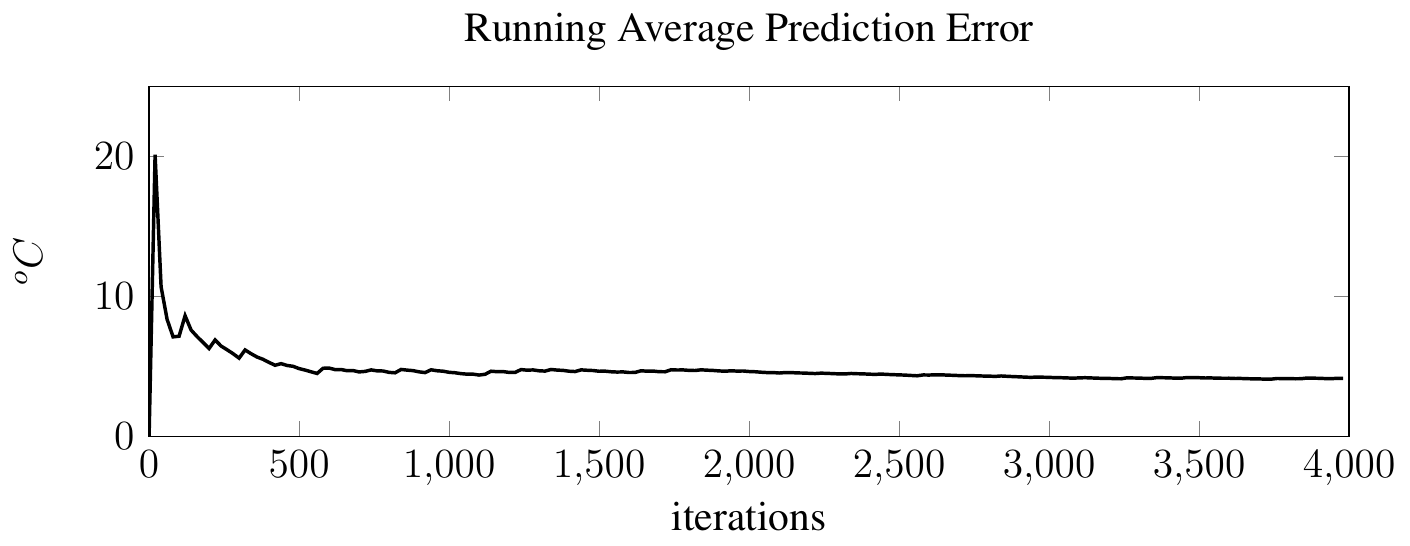}
\fi
\caption{Running average prediction error.}
\label{fig:RunningAveragePredictionError}
\end{figure}

\begin{figure}[t!]
\centering
\iffigures
\includegraphics[scale=0.8]{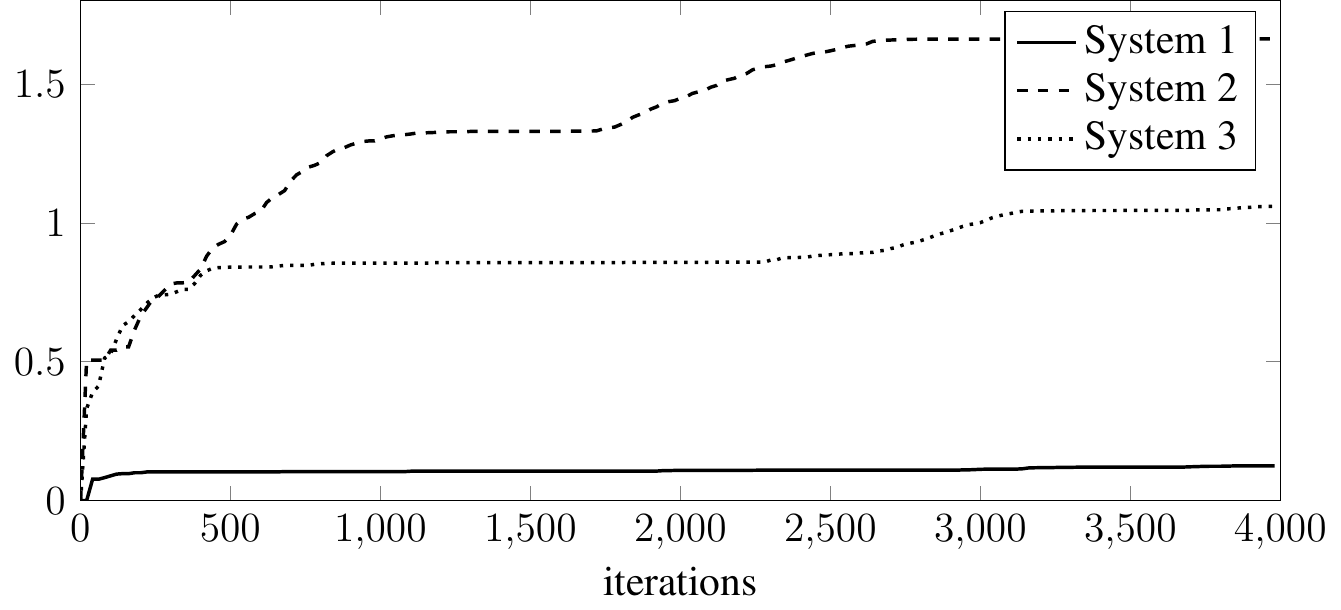}
\fi
\caption{Running average performance for partition pattern 1.}
\label{fig:RunningAveragePerformance_Partition1}
\end{figure}

\begin{figure}[t!]
\centering
\iffigures
\includegraphics[scale=0.8]{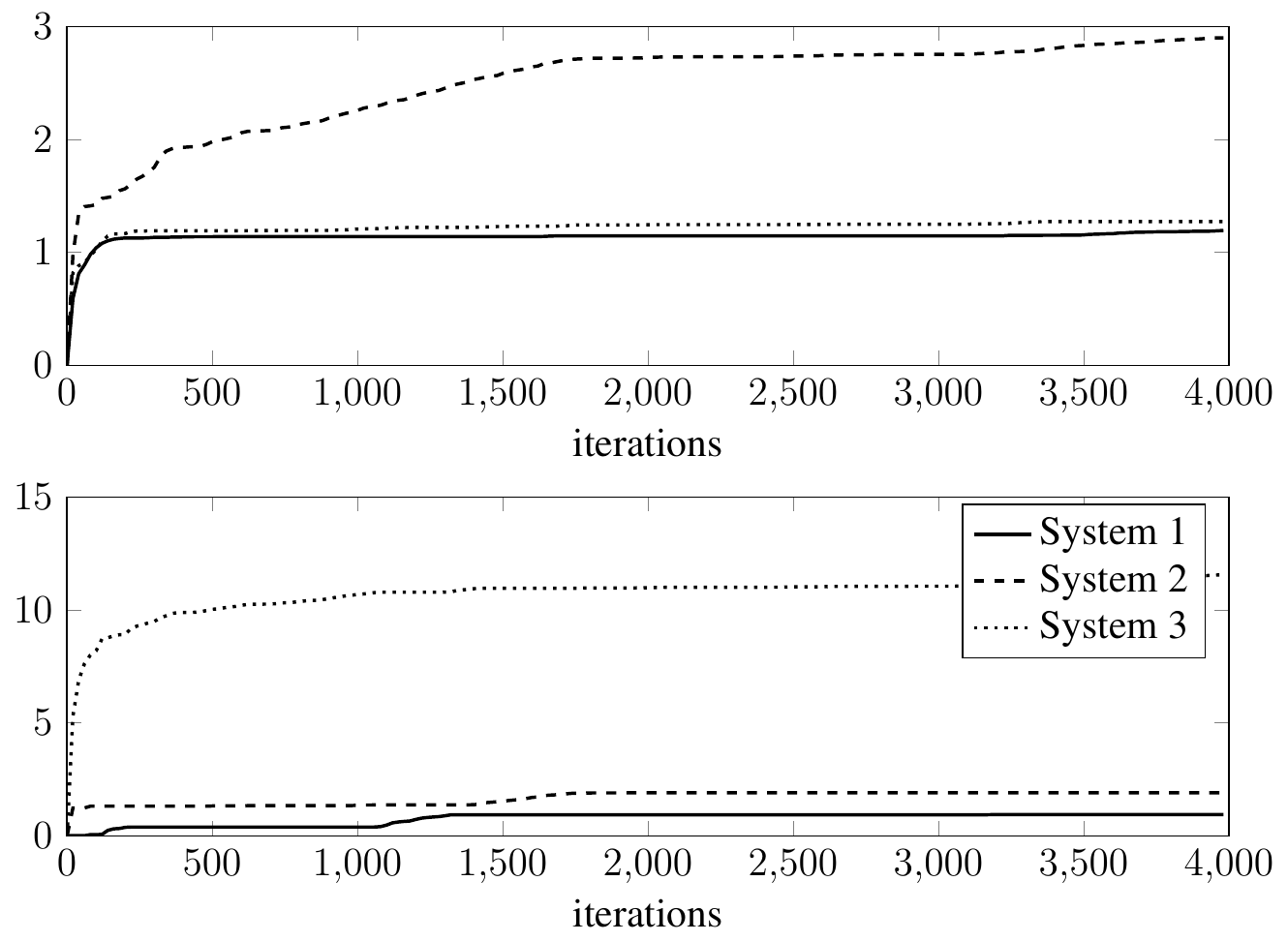}
\fi
\caption{Running average performance for partition pattern 2.}
\label{fig:RunningAveragePerformance_Partition2}
\end{figure}

We set the sampling period equal to $T_s=\nicefrac{1}{12}{\rm h}$ and the period of the supervisory process was set equal to $T_e=400T_s$. The step-size of the dynamics was set equal to $\epsilon=0.05$ and the perturbation probability was set equal to $\lambda=0.03$. Figure~\ref{fig:ProbabilitiesOverPartitions} shows how the probabilities over the considered two partition-pattern profiles evolve with time, while Figures~\ref{fig:ProbabilitiesOverModelsForPartition1}--\ref{fig:ProbabilitiesOverModelsForPartition2} show how the probabilities over the available models evolve in case of the two partition schemes, respectively. Figure~\ref{fig:RunningAveragePredictionError} shows the running average prediction error, while Figures~\ref{fig:RunningAveragePerformance_Partition1}--\ref{fig:RunningAveragePerformance_Partition2} show the running average performance of the models.

The supervisory process follows the (\emph{in probability}) type of convergence of Theorem~\ref{Th:GlobalConvergence}. See, for example, Figure~\ref{fig:ProbabilitiesOverPartitions}, where the process switches between the two partition schemes (resulting in variations in their probability distribution). The partition which provides the highest performance will eventually be picked more often (as it is the case for partition profile 1). 

A similar behavior is also observed in Figure~\ref{fig:ProbabilitiesOverModelsForPartition1} (for model selection), where the process selects Model~2 for a quite large fraction of the implementation time (which agrees with the observation that Model~2 provides the highest running average performance, Figure~\ref{fig:RunningAveragePerformance_Partition1}). As Figure~\ref{fig:ProbabilitiesOverModelsForPartition1} demonstrates, all models are selected frequently, thus if a model starts performing better at later times, then the supervisory scheme will adapt to the new best option. 

Finally, it is worth mentioning that Figures~\ref{fig:RunningAveragePerformance_Partition1}--\ref{fig:RunningAveragePerformance_Partition2} constitute an informal verification of Assumptions~\ref{As:MeanAsymptoticPerformance}, \ref{As:DistinctPerformances} and \ref{As:BestPartitionModels}. In particular, there is a model which provides better performance than any other selection (Assumption~\ref{As:BestPartitionModels}). Furthermore, the running average performance of all models is convergent when the models have been selected sufficiently often (Assumption~\ref{As:MeanAsymptoticPerformance}). Finally, each model provides a distinct performance in the mean (Assumption~\ref{As:DistinctPerformances}).

\subsection{Discussion}

The \emph{in-probability} type of convergence is a consequence of making choices based on strategy vectors and using a constant step-size $\epsilon>0$. It may result in suboptimal selections over (arbitrarily) small portions of the implementation time, however it provides a highly attractive feature for such supervisory schemes, that is its \emph{linear} computational complexity. 

Note that we may consider dynamics that provide stronger convergence guarantees. However, this most likely will result in losing the linear complexity feature. To see this, consider instead an algorithm that directly addresses optimization (\ref{eq:InitialOptimizationProblem}): at each iteration $k$, \emph{all} partition-patterns in $\cP$ and \emph{all} available models in $\cM$ are being evaluated and compared with each other. In this case, at each iteration $k$, we may select the partition/model that have provided the best performance so far. Such dynamics (which belong to minimizing-regret type of dynamics) will provide stronger convergence guarantees. However, they require a number of computations per iteration that is proportional to $(c\magn{\cM}+1)\magn{\cP}$ (without taking into account the required training of each available model in each one of the subsets of the available partitions). Depending on the size of the sets $\cM$ and $\cP$, such computational complexity may be considerable larger compared to the one of the proposed dynamics, $c\magn{\cM}+\magn{\cP}$.

%%%%%%%%%%%%%%%%%%%%%%%%%%%%%%%%%%%%%%%%%%%%%%%%%%%%%%%%%%%%%%%%%%%%%%%%%%
\section{Technical Derivation}		\label{sec:TechnicalDerivation}

The following analysis is concerned with the asymptotic behavior of the strategies $\vv_{\vpq,k}$ and $\vw_k$. Let us define: % the following reward functions: 
\begin{equation*}
\mathbf{r}_{\vpq,k} \df \sum_{i\in\cM}\ve_i r_{\vpq,k}(i) \in \mathbb{R}_{+}^{\magn{\cM}},
\end{equation*} 
\begin{equation*}
\overline{R}_{k}(\vp,\vv) \df \sum_{\alpha\in\Ap} \sum_{i\in\cM}r_{\vpq,k}(i)\vv_{\vpq}[i] \in\mathbb{R}_{+},
\end{equation*}
\begin{equation*}
\overline{\mathbf{R}}_{k}(\vv) \df \sum_{\vp\in\cP} \ve_{\vp} \overline{R}_{k}(\vp,\vv) \in \mathbb{R}_{+}^{\magn{\cP}}.
\end{equation*}
Note that $\mathbf{r}_{\vpq,k}$ denotes the vector of performances realized by each model in the partition-pattern profile $\vp$ and subset $\alpha$, i.e., $\mathbf{r}_{\vpq,k}[i] \equiv r_{\vpq,k}(i)$. Moreover, $\overline{R}_k(\vp,\vv)$ denotes the expected performance of partition-pattern profile $\vp$ given the strategies $\vv_{\vpq}$ applied to each subset $\alpha$; and, $\overline{\mathbf{R}}_k(\vv)$ denotes the vector of the expected performances exhibited by each partition-pattern profile. 

The corresponding limits as $k\to\infty$ will be denoted by $\mathbf{r}_{\vpq,\infty}$, $\overline{R}_{\infty}(\vp,\vv)$ and $\overline{\mathbf{R}}_{\infty}(\vv)$, respectively. 

%---------------------------------------
\subsection{ODE approximation}

The convergence properties of the supervisory scheme can be analyzed via the ODE method for stochastic approximations \cite{KushnerYin03}. It is based on the selection of an \emph{Ordinary Differential Equation}, the limit sets of which can be associated with the observed asymptotic behavior of recursions (\ref{eq:ModelsStrategiesUpdatesB}), (\ref{eq:PartitionsStrategyUpdates}). 

% We consider again the canonical path space $\Omega$ and the $\sigma$-algebra $\mathfrak{F}_n$ defined in \ref{sec:Assumptions}. 

The convergence properties of (\ref{eq:ModelsStrategiesUpdatesB}), (\ref{eq:PartitionsStrategyUpdates}) will be associated with the limit sets of the following system of ODE's:
\begin{subequations}	\label{eq:ODE_Approximation}
\begin{eqnarray}
\dot{\vv}_{\vpq} & = & \overline{\mathbf{g}}_{\vpq,\infty}^{\lambda}(\vv_{\vpq}), \label{eq:ODE_Approximation_v} \\
\dot{\vw} & = & \overline{\mathbf{f}}_{\infty}^{\lambda}(\vv,\vw) \label{eq:ODE_Approximation_w}
\end{eqnarray}
\end{subequations}
for each $\vp\in\cP$ and $\alpha\in\Ap$, where 
\begin{eqnarray*}
\lefteqn{\overline{\mathbf{g}}_{\vpq,\infty}^{\lambda}(\vv_{\vpq}) \df } \cr && \mathbb{E}_{\mathfrak{F}_k}\left[\left. r_{\vpq,\infty}(S_{\vpq,k})\left[\ve_{S_{\vpq,k}} - \vv_{\vpq,k}\right]\right|\vv_{\vpq,k}=\vv_{\vpq}\right],
\end{eqnarray*}
\begin{eqnarray*}
\overline{\mathbf{f}}_{\infty}^{\lambda}(\vv,\vw) \df
\mathbb{E}_{\mathfrak{F}_k}\left[\left. R_{\infty}(\vp_k,\mathbf{S}_k) \left[\ve_{\vp_k} - \vw_k\right]\right|\vv_k=\vv,\vw_k=\vw\right].
\end{eqnarray*}
In the unperturbed case ($\lambda=0$), according to \cite{ChasparisShammaRantzer14_IJGT}, the above ODE takes on a special form:
$\overline{\mathbf{g}}_{\vpq,\infty}^{0}(\vv_{\vpq})\equiv \mathbf{V}(\vv_{\vpq})\mathbf{r}_{\vpq,\infty}$ where
$\mathbf{V}:\mathbf{\Delta}(\magn{\cM})\mapsto\mathbb{R}^{\magn{\cM}\times\magn{\cM}}$ is such that 
\begin{eqnarray*}
[\mathbf{V}(\vv_{\vpq})]_{ji} \df \begin{cases}
\vv_{\vpq}[j]\cdot (1-\vv_{\vpq}[j]), & j=i \\
 -\vv_{\vpq}[j]\cdot \vv_{\vpq}[i], & j\neq{i}
\end{cases}, \quad \forall i,j\in\cM,
\end{eqnarray*}
and 
$\overline{\mathbf{f}}_{\infty}^{0}(\vv,\vw)\equiv \mathbf{W}(\vw)\overline{\mathbf{R}}_{\infty}(\vv),$ where $\mathbf{W}:\mathbf{\Delta}(\magn{\cP})\mapsto\mathbb{R}^{\magn{\cP}\times\magn{\cP}}$ is such that 
\begin{eqnarray*}
[\mathbf{W}(\vw)]_{ji} \df \begin{cases}
\vw[j]\cdot (1-\vw[j]), & j=i \\
-\vw[j]\cdot \vw[i], & j\neq{i}
\end{cases} \quad \forall j,i\in\cP.
\end{eqnarray*}

The connection between the ODE~(\ref{eq:ODE_Approximation}) for some $\lambda>0$ and the discrete-time recursions (\ref{eq:ModelsStrategiesUpdatesB}), (\ref{eq:PartitionsStrategyUpdates}) is established by the following proposition.
\begin{proposition}[Asymptotic properties]	\label{Pr:AsymptoticProperties}
For some $\lambda>0$, let $\mathcal{L}$ denote the limit points\footnote{The set of \emph{limit points} $\mathcal{L}$ of an ODE $\dot{\mathbf{x}}=\mathbf{g}(\mathbf{x})$ with domain $A$ is defined as $\mathcal{L}\df \lim_{t\to\infty}\bigcup_{x\in{A}}\{\mathbf{x}(s),s\geq{t}:\mathbf{x}(0)=x\}$, i.e., it is the set of all points in $A$ to which the solution of the ODE converges.} of the system of ODE's~(\ref{eq:ODE_Approximation_v})--(\ref{eq:ODE_Approximation_w}). For any $\delta>0$, the fraction of time that the linear-time interpolation\footnote{The linear-time interpolation of a recursion $\mathbf{x}_k$, $k=0,1,...$, is defined as $\overline{\mathbf{x}}(\tau)=\mathbf{x}_k$ for all $\epsilon k\leq{\tau}<\epsilon(k+1)$.} of the recursions $\{\vv_{\vpq,k}\}_{\vpq}$ and $\vw_k$ spends in $\mathcal{N}_{\delta}(\mathcal{L})$ goes to one (in probability) as $\epsilon\to{0}$ and $k\to\infty$.
\end{proposition}
\begin{proof}
% See Appendix~\ref{Ap:ProofAsymptoticProperties}.
The proof follows directly from Theorem~8.4.1 in \cite{KushnerYin03}, due to a) Assumption~\ref{As:MeanAsymptoticPerformance}, b) the performance functions are uniformly bounded, c) the functions $\overline{\mathbf{g}}_{\vpq,\infty}^{\lambda}(\cdot)$ and $\overline{\mathbf{f}}_{\infty}^{\lambda}(\cdot,\cdot)$ are continuous in their domain.
\end{proof}

In other words, Proposition~\ref{Pr:AsymptoticProperties} states that if we consider a sufficiently large number of iterations $k$, the fraction of time that the stochastic recursion will spend in $\mathcal{N}_{\delta}(\mathcal{L})$ goes to one (in probability) as $\epsilon\to{0}$ and $k\to\infty$. Potential limit points of the ODE's~(\ref{eq:ODE_Approximation_v})--(\ref{eq:ODE_Approximation_w}) are its stationary points, i.e., points $(\vv^*,\vw^*)$ at which $\overline{\mathbf{g}}_{\vpq,\infty}^{\lambda}(\vv_{\vpq}^*) = 0$, and $\overline{\mathbf{f}}_{\infty}^{\lambda}(\vv^*,\vw^*)=0$. 

\begin{proposition}[Unique stationary point]
Let $(\vv^*,\vw^*)$ be a pure strategy pair corresponding to the best selection $\vp^*$ and $\vS^*$. For sufficiently small $\lambda>0$, there exists $\delta=\delta(\lambda)$ with $\delta(\lambda)\to{0}$ as $\lambda\to{0}$, such that $\mathcal{N}_{\delta}(\vv^*,\vw^*)$ contains the unique stationary point of the ODE's~(\ref{eq:ODE_Approximation_v})--(\ref{eq:ODE_Approximation_w}).
\end{proposition}
\begin{proof}
The statement follows directly from \cite[Proposition~3.5]{ChasparisShamma11_DGA} and the fact that for the perturbed ODE~(\ref{eq:ODE_Approximation}), the partition and combination of models corresponding to $(\vv^*,\vw^*)$ provides the highest reward compared to any other pair selection. Note further (also according to \cite[Proposition~3.5]{ChasparisShamma11_DGA}), that any other pure strategy (which is a stationary point under the unperturbed dynamics) no longer defines a stationary point for the perturbed dynamics.
\end{proof}

%---------------------------------------
\subsection{Proof of Theorem~\ref{Th:GlobalConvergence}}

Let us define the set of pure strategy profiles by $\mathbf{\Delta}^*$. Let also $(\vv^*,\vw^*)\in\mathbf{\Delta}^*$ correspond to the best selection of Assumption~\ref{As:BestPartitionModels}. For some $\delta>0$, let $\mathcal{D}_{\delta}\df(\mathcal{V}\times\mathcal{W})\backslash\mathcal{N}_{\delta}(\mathbf{\Delta}^*),$ which includes all strategy pairs but a small neighborhood about $\mathbf{\Delta}^*$. By Assumption~\ref{As:BestPartitionModels}, any pair $(\vv,\vw)\in\mathcal{D}_{\delta}$ satisfies ${\vw^*}\tr\overline{\mathbf{R}}_{\infty}(\vv^*) > \vw\tr\overline{\mathbf{R}}_{\infty}(\vv)$. This is because, (a) $\overline{R}_{\infty}(\vp^*,\vv^*) > \overline{R}_{\infty}(\vp,\vv^*)$ (according to (H2) of Assumption~\ref{As:BestPartitionModels}), and (b) $\overline{R}_{\infty}(\vp,\vv^*) > \overline{R}_{\infty}(\vp,\vv)$ (according to (H1) condition of Assumption~\ref{As:BestPartitionModels}). In particular,
\begin{eqnarray*}
\begin{array}{rll}
{\vw^{*}}\tr\overline{\mathbf{R}}_{\infty}(\vv^*) & = \overline{R}_{\infty}(\vp^*,\vv^*) & % & \mbox{(due to $\vw^*\equiv \ve_{\vp^*}$)} 
\cr
& > \vw\tr\sum_{\vp\in\cP}\ve_{\vp}\overline{R}_{\infty}(\vp,\vv^*) & % \quad \mbox{ for all } \vw\neq \vw^*  & \mbox{(due to condition (H2))} 
\cr
& > \vw\tr\sum_{\vp\in\cP}\ve_{\vp}\overline{R}_{\infty}(\vp,\vv) & = \vw\tr\overline{\mathbf{R}}_{\infty}(\vv). % \quad \mbox{ for all } \vv\neq \vv^*  & \mbox{(due to condition (H1))} 
%\cr
%& = & \vw\tr\overline{\mathbf{R}}_{\infty}(\vv). % & 
\end{array}
\end{eqnarray*}
for all $\vw\neq \vw^*$ and due to $\vw^*\equiv \ve_{\vp^*}$ and conditions (H1), (H2).
Thus, ${\vw^*}\tr\overline{\mathbf{R}}_{\infty}(\vv^*) > \vw\tr\overline{\mathbf{R}}_{\infty}(\vv)$ for all $(\vv,\vw)\in\mathcal{D}_{\delta}$.

Define the nonnegative function $V:\mathcal{D}_{\delta}\mapsto[0,\infty),$ such that: $$V(\vv,\vw) \df {\vw^*}\tr \overline{\mathbf{R}}_{\infty}(\vv^*) - \vw\tr\overline{\mathbf{R}}_{\infty}(\vv).$$ Note that $\nabla_{\vv_{\vpq}}V(\vv,\vw) = - \vw_{\vp} \mathbf{r}_{\vpq,\infty}\tr,$  and $\nabla_{\vw}V(\vv,\vw) = -\overline{\mathbf{R}}_{\infty}(\vv)\tr.$ Thus, its time derivative satisfies:
\begin{eqnarray*}
\dot{V}(\vv,\vw) & = & -\sum_{\vp\in\cP}\vw_{\vp}\sum_{\alpha\in\Ap}\mathbf{r}_{\vpq,\infty}\tr \mathbf{V}(\vv_{\vpq})\mathbf{r}_{\vpq,\infty} - \cr && \overline{\mathbf{R}}_{\infty}(\vv)\tr\mathbf{W}(\vw)\overline{\mathbf{R}}_{\infty}(\vv) + O(\lambda),
\end{eqnarray*}
where $O(\lambda)$ denotes terms of order of $\lambda$. Note further that according to the definition of the matrices $\mathbf{V}$ and $\mathbf{W}$,\footnote{We abuse notation here by using $i$, $j$ as both the indices and the elements of the corresponding sets $\cM$ and $\cP$.}
\begin{eqnarray*}
\lefteqn{\mathbf{r}_{\vpq,\infty}\tr \mathbf{V}(\vv_{\vpq})\mathbf{r}_{\vpq,\infty} = }\cr && \sum_{i=1}^{\magn{\cM}}\sum_{j=1,j>i}^{\magn{\cM}}\vv_{\vpq}[i]\vv_{\vpq}[j]\left(r_{\vpq,\infty}(i)-r_{\vpq,\infty}(j)\right)^2 > 0,
\end{eqnarray*}
\begin{eqnarray*}
\lefteqn{\overline{\mathbf{R}}_{\infty}(\vv)\tr\mathbf{W}(\vw)\overline{\mathbf{R}}_{\infty}(\vv) =}\cr && \sum_{i=1}^{\magn{\cP}}\sum_{j=1,j>i}^{\magn{\cP}}\vw[i]\vw[j]\left(\overline{R}_{\infty}(i,\vv) - \overline{R}_{\infty}(j,\vv)\right)^2 > 0,
\end{eqnarray*}
for all $(\vv,\vw)\in\mathcal{D}_{\delta}$, since (a) $\vv_{\vpq}[i]\in(0,1)$, for all $\vp\in\cP$, $\alpha\in\Ap$ and $i\in\cM$, (b) $\vw[i]\in(0,1)$ for all $i\in\cP$, and (c) $r_{\vpq,\infty}(i)\neq r_{\vpq,\infty}(j)$ for all $i,j\in\cM$, $i\neq{j}$, according to Assumption~\ref{As:DistinctPerformances}. Thus, if we take $\lambda$ sufficiently small, $\dot{V}(\vv,\vw)<0$ for all $(\vv,\vw)\in\mathcal{D}_{\delta}$. Furthermore, since $V(\vv,\vw)>{0}$ in $\mathcal{D}_{\delta}$ we conclude that the set $\mathcal{N}_{\delta}(\mathbf{\Delta}^*)$ is globally asymptotically stable in the sense of Lyapunov. 

It remains to investigate convergence to the set of pure strategy profiles within $\mathcal{N}(\mathbf{\Delta}^*)$. 

(A) We first exclude convergence from any pure strategy profile in $\mathbf{\Delta}^*$ which does not correspond to the best selection $(\vv^*,\vw^*)$. Let us consider the vector field $\overline{\mathbf{g}}_{\vpq,\infty}^{\lambda}(\vv_{\vpq})$ in $\mathcal{N}_{\delta}(\vv',\vw')$ for some $(\vv',\vw')\in\mathbf{\Delta}^*$, i.e., $(\vv',\vw')$ corresponds to a pure strategy profile. The $i$th entry of this vector field, $i\in\{1,...,\magn{\cM}\}$, evaluated at some strategy profile $(\vv,\vw)\in\mathcal{N}_{\delta}(\vv',\vw')$ takes on the following form, 
\begin{eqnarray*}
\lefteqn{\overline{\mathbf{g}}_{\vpq,\infty}^{\lambda}(\vv_{\vpq})[i]} \cr & = & 
r_{\vpq,\infty}(i)\Big((1-\lambda)\vv_{\vpq}[i] + \frac{\lambda}{\magn{\cM}}\Big) - \cr && \sum_{j\in\cM}r_{\vpq,\infty}(j)\Big((1-\lambda)\vv_{\vpq}[j]+\frac{\lambda}{\magn{\cM}}\Big)\vv_{\vpq}[i]
\cr & \approx &  \Big(r_{\vpq,\infty}(i)-r_{\vpq,\infty}(i')\vv_{\vpq}[i']\Big)\vv_{\vpq}[i] + r_{\vpq,\infty}(i)\frac{\lambda}{\magn{\cM}}
\end{eqnarray*}
plus terms of order $\lambda\delta$ and $\delta^2$, for all $i\neq{i'}$, where $i'$ corresponds to the model in $\cM$ assigned probability one in the vector $\vv_{\vpq}'$. Let us take $\vv_{\vpq}'\neq\vv_{\vpq}^*$, i.e., $\vv_{\vpq}'$ does not correspond to the best selection. Evaluate the above expression at $i^*$ which corresponds to the model in $\cM$ assigned probability one in $\vv_{\vpq}^*$. Then, we have
\begin{eqnarray*}
\lefteqn{\overline{\mathbf{g}}_{\vpq,\infty}^{\lambda}(\vv_{\vpq})[i^*] \approx }\cr && \Big(r_{\vpq,\infty}(i^*)-r_{\vpq,\infty}(i')\vv_{\vpq}[i']\Big)\vv_{\vpq}[i^*] + r_{\vpq,\infty}(i^*)\frac{\lambda}{\magn{\cM}}
\end{eqnarray*}
plus terms of order $\lambda\delta$ and $\delta^2$. Given that $\vv_{\vpq}$ is taken within $\mathcal{N}_{\delta}(\vv',\vw')$, and the fact that $r_{\vpq,\infty}(i^*)>r_{\vpq,\infty}(i')>0$, there exists $\delta^*>0$ such that, for any $\delta<\delta^*$, we have: 
\begin{eqnarray*}
\Big(r_{\vpq,\infty}(i^*)-r_{\vpq,\infty}(i')\vv_{\vpq}[i']\Big)\vv_{\vpq}[i^*] + r_{\vpq,\infty}(i^*)\frac{\lambda}{\magn{\cM}} \geq \cr  \Big(r_{\vpq,\infty}(i^*)-r_{\vpq,\infty}(i')\Big)\vv_{\vpq}[i^*] + r_{\vpq,\infty}(i^*)\frac{\lambda}{\magn{\cM}} > 0.
\end{eqnarray*}
We conclude that, if we take $\lambda$ sufficiently small, the above expression strictly dominates terms of order $\lambda\delta$ and $\delta^2$, implying that $\overline{\mathbf{g}}_{\vpq,\infty}^{\lambda}(\vv_{\vpq})[i^*]>0$ within $\mathcal{N}_{\delta}(\vv',\vw')$, i.e., the vector field points towards the exterior of $\mathcal{N}_{\delta}(\vv',\vw')$. Thus, the set $\mathcal{N}_{\delta}(\vv',\vw')$ belongs to the region of attraction of $\mathcal{D}_{\delta}$ when we take $\delta$ and $\lambda$ sufficiently small. In other words, the set $\mathcal{N}_{\delta}(\vv^*,\vw^*)$ is globally asymptotically stable. % in the sense of Lyapunov.

(B) It remains to investigate the vector field $\overline{\mathbf{g}}_{\vpq,\infty}^{\lambda}(\vv_{\vpq})$ within $\mathcal{N}_{\delta}(\vv^*,\vw^*)$, i.e., within a small neighborhood of the best selection. By directly implementing Proposition~3.6 of \cite{ChasparisShamma11_DGA} to ODE's (\ref{eq:ODE_Approximation_v})--(\ref{eq:ODE_Approximation_w}), we have that there exists a locally asymptotically stable stationary point $(\tilde{\vv}^{\lambda},\tilde{\vw}^{\lambda})$ satisfying $\tilde{\vv}^{\lambda}\to\vv^*$ and $\tilde{\vw}^{\lambda}\to\vw^*$ as $\lambda\to{0}$.
%\begin{itemize}
%\item[(a)] For sufficiently small $\lambda>0$, the strategy $\tilde{\vv}_{\lambda}$ is a locally asymptotically stable stationary point of the ODE~(\ref{eq:ODE_Approximation_v}). This follows directly from Proposition~3.6 of \cite{ChasparisShamma11_DGA} given that (i) $\tilde{\vv}_{\lambda}$ is a stationary point of the ODE~(\ref{eq:ODE_Approximation_v}) according to Proposition~\ref{Pr:StationaryPoints}, and (ii) strategy $\tilde{\vv}_{\lambda}$ corresponds to a $\lambda$ perturbation of the best selection $\vv^*$. (Note that this conclusion holds independently of the value of $\vw$.)
%\item[(b)] For sufficiently small $\lambda>0$, the strategy $\tilde{\vw}_{\lambda}$ is a locally asymptotically stable stationary point of the ODE~(\ref{eq:ODE_Approximation_w}) when evaluated at $\tilde{\vv}_{\lambda}$. This follows directly from Proposition~3.6 of \cite{ChasparisShamma11_DGA} given that (i) $\tilde{\vw}_{\lambda}$ is a stationary point of the ODE~(\ref{eq:ODE_Approximation_w}) when evaluated at $\tilde{\vv}_{\lambda}$, according to Proposition~\ref{Pr:StationaryPoints}, and (ii) strategy $\tilde{\vw}_{\lambda}$ corresponds to a $\lambda$ perturbation of the best selection $\vw^*$ when the model selection follows $\vv^*$.
%\end{itemize}
%Overall, for sufficiently small $\lambda>0$, $(\tilde{\vv}_{\lambda},\tilde{\vw}_{\lambda})$ is a locally asymptotically stable stationary point of the ODE~(\ref{eq:ODE_Approximation}). 
Thus, if we pick $\lambda$ sufficiently small and $\delta=\delta(\lambda)$ such that $\mathcal{N}_{\delta}(\vv^*,\vw^*)$ belongs to the region of attraction of $(\tilde{\vv}_{\lambda},\tilde{\vw}_{\lambda})$, then $(\tilde{\vv}_{\lambda},\tilde{\vw}_{\lambda})$ will be the unique limit point within $\mathcal{N}_{\delta}(\vv^*,\vw^*)$.

From (A) and (B), we conclude that for sufficiently small $\lambda$, $(\tilde{\vv}_{\lambda},\tilde{\vw}_{\lambda})$ is globally asymptotically stable. Thus, according to Proposition~\ref{Pr:AsymptoticProperties}, the conclusion follows.

%%%%%%%%%%%%%%%%%%%%%%%%%%%%%%%%%%%%%%%%%%%%%%%%%%%%%%%%%%%%%%%%%%%%%%%%%%%%%%%%%%%%%%%%%%%%%%%%%%%%%
\section{Conclusions and Future Work} \label{sec:Conclusions}

We introduced a supervisory online prediction scheme specifically tailored for bilinear systems, which incorporates two parallel decision processes: a) the selection of a partition-pattern profile for the inputs' domain, and b) the selection of a model profile for each partition-pattern profile. We showed that convergence is attained (in probability) to the partition and model profile with the smallest prediction error. Such supervisory scheme is appropriate when the operating conditions might change with time, as in the case of the thermal dynamics in buildings, and when predictions need to be provided online as in standard model predictive control formulations, thus automating the process of prediction model selection.

%%%%%%%%%%%%%%%%%%%%%%%%%%%%%%%%%%%%%%%%%%%%%%%%%%%%%%%%%%%%%%%%%%%%%%%%%
% Appendices
%%%%%%%%%%%%%%%%%%%%%%%%%%%%%%%%%%%%%%%%%%%%%%%%%%%%%%%%%%%%%%%%%%%%%%%%%

% \appendices

%%%%%%%%%%%%%%%%%%%%%%%%%%%%%%%%%%%%%%%%%%%%%%%%%%%%%%%%%%%%%%%%%%%%%%%%%
% Acknowledgements
%%%%%%%%%%%%%%%%%%%%%%%%%%%%%%%%%%%%%%%%%%%%%%%%%%%%%%%%%%%%%%%%%%%%%%%%%

%\section*{Acknowledgment}

%%%%%%%%%%%%%%%%%%%%%%%%%%%%%%%%%%%%%%%%%%%%%%%%%%%%%%%%%%%%%%%%%%%%%%%%%%
%% Bibliography
%%%%%%%%%%%%%%%%%%%%%%%%%%%%%%%%%%%%%%%%%%%%%%%%%%%%%%%%%%%%%%%%%%%%%%%%%%

% \bibliographystyle{elsarticle-harv} % IEEEtran
\bibliographystyle{IEEETran}
\bibliography{Bibliography}

\end{document}